\documentclass[12pt]{article}
\usepackage{amsmath,amssymb,theorem,cite,epsfig,url,psfrag,eepic}
\advance \topmargin by -\headheight
\advance \topmargin by -\headsep     
\evensidemargin \oddsidemargin
\def\Maketitle{{\def\newpage{}\maketitle}}
\makeatletter
\def\Appendix{\appendix
  \def\@seccntformat##1{Appendix~\csname the##1\endcsname.~~}}
\makeatother
\makeatletter
\@addtoreset{equation}{section}

\makeatother
\theorembodyfont{\sffamily}
\newtheorem{theorem}{Theorem}[section]

\newtheorem{proposition}[theorem]{Proposition}

\newenvironment{proof}[1][Proof]{\begin{trivlist}
\item[\hskip \labelsep {\bfseries #1}]}{\end{trivlist}}
\def\jac{\mathop{\rm Jac}\nolimits}
\begin{document}
\title{\textbf{Parafermionic polynomials,\\ Selberg integrals and three-point correlation\\ function in parafermionic Liouville field theory}\vspace*{12pt}}
\vspace*{12pt}
\author{M.~A.~Bershtein$^{1,2}$, V.~A.~Fateev$^{1,3}$ and A.~V.~Litvinov$^{1,4}$\\[\medskipamount]
$^1$~\parbox[t]{0.85\textwidth}{\normalsize\it\raggedright
Landau Institute for Theoretical Physics,
142432 Chernogolovka, Russia}\vspace*{2pt}\\[\medskipamount]
$^2$~\parbox[t]{0.85\textwidth}{\normalsize\it\raggedright
Independent University of Moscow, 11 Bolshoy Vlasyevsky pereulok,\\ 119002 Moscow, Russia}\vspace*{5pt}\\[\medskipamount]
$^3$~\parbox[t]{0.85\textwidth}{\normalsize\it\raggedright
Laboratoire de Physique Th\'eorique et Astroparticules, UMR5207 CNRS-UM2, Universit\'e
Montpellier~II, Pl.~E.~Bataillon, 34095 Montpellier, France}\vspace*{5pt}\\[\medskipamount]
$^4$~\parbox[t]{0.85\textwidth}{\normalsize\it\raggedright
~NHETC, Department of Physics and Astronomy, Rutgers University,\\ 136 Frelinghuysen Road, Piscataway, NJ 08855-0849, U.S.A.}}
\rightline{RUNHETC-2010-2}
\date{}
\Maketitle
\begin{abstract}
 In this paper we consider parafermionic Liouville field theory. We study integral representations of three-point correlation  functions and develop a method allowing us to compute them exactly. In particular, we evaluate the generalization of Selberg integral obtained by insertion of parafermionic polynomial. Our result is justified by different approach based on dual representation of parafermionic Liouville field theory described by three-exponential model.
\end{abstract}
\tableofcontents
\section{Introduction}\label{Introduction}
Non-rational conformal field theories attract a lot of attention in contemporary theoretical physics since Polyakov's suggestion to study strings in non-critical dimension \cite{Polyakov:1981rd}. In Polyakov's approach metric on a surface of the string due to conformal anomaly become dynamical and its behavior is described by Liouville field theory  which is the most well known and widely studied in the literature example of non-rational CFT \cite{Dorn:1992at,Dorn:1994xn,Zamolodchikov:1995aa,Teschner:2001rv}. An important question which one usually asks about any CFT is the complete set of correlation functions of local fields. The conformal symmetry of the theory plus the hypothesis about the operator algebra help to answer this question very much \cite{Belavin:1984vu}. In particular, the problem can be reduced to the problem of studying of the operator algebra of so called primary fields. In \cite{Belavin:1984vu} the set of conformal theories so called minimal models has been constructed. The main property of these theories is a  finite number of primary fields. Contrary, the main ``disadvantage'' of non-rational CFT -- infinite (even uncountable) number of primary fields was probably the main reason why the progress in this direction has been postponed. In \cite{Dorn:1992at,Dorn:1994xn,Zamolodchikov:1995aa,Teschner:2003en} the three-point correlation functions or equivalently the structure constants of the operator algebra of primary fields where found exactly (DOZZ formula). Conceptually, this important result can be viewed as a complete solution of Liouville field theory because it allows to find (in principle) any multipoint correlation function.  Another piece of information about the theory which is needed for the construction of all correlation functions is encoded in the conformal blocks which are objects defined only by the conformal symmetry. We note that the conformal blocks with general external and intermediate conformal dimensions appear essentially only in non-rational CFT. Despite the fact that there is a straightforward algebraic procedure allowing  to construct the conformal blocks their explicit determination is a tedious problem which become actual again since the discovery of the connection between non-rational CFT and four-dimensional $\mathcal{N}=2$ supersymmetric gauge theories in recent paper of Alday, Gaiotto and Tachikawa \cite{Alday:2009aq} (AGT relation). Probably the best way to compute the expansion of the conformal blocks is the recursion formula suggested by Alyosha Zamolodchikov \cite{Zamolodchikov:1985ie} (we note that AGT relation also gives very effective algorithm allowing to compute the conformal blocks).

So far, only Liouville field theory as well as its $\mathfrak{sl}(n)$ generalizations called Toda field theories   \cite{Fateev:2007ab,Fateev:2008bm,Wyllard:2009hg} were studied in the context of AGT relation. However, we still have many examples of non-rational CFT's whose role in gauge theory--non-rational CFT duality is unclear at present. One possible candidate is the supersymmetric generalization of Liouville field theory  which appears essentially in  Polyakov's approach applied to superstring in non-critical dimension \cite{Polyakov:1981re}. This theory is now studied in the same extent as ordinary bosonic LFT.  In particular, the supersymmetric generalization of DOZZ formula is known \cite{Rashkov:1996jx,Poghosian:1996dw} as well as the recursion formula for the conformal blocks \cite{Belavin:2007gz,Hadasz:2007nt,Hadasz:2008dt}. Possible generalization of superstring theory called fractional superstring theory was suggested in \cite{Argyres:1990aq} which is roughly speaking based on idea to substitute ordinary fermions by $Z_{N}$ parafermions \cite{Fateev:1985mm}. 

Theories with parafermionic symmetry were studied in different contexts (see papers citing \cite{Fateev:1985mm})  and their massive deformations in \cite{Koberle:1979sg,Fateev:1990bf,Fateev:2006js,Fateev:2008zz,Fateev:2009kp}. We give an overview of $Z_{N}$ parafermionic CFT in appendix \ref{PF}. Probably, the most promising application of this CFT appears in studying of non-Abelian Fractional Quantum Hall Effect \cite{Read:1998ed}. So called Read-Rezayi state constructed in \cite{Read:1998ed} gives trial ground state wavefunction  for the system of charged particles on a sphere of certain radius in a uniform radial magnetic field at Landau level filling $N/(2+NM)$ where $M$ is an integer number. Up to some standard factors the $(N,2)$ Read-Rezayi state is proportional to the multipoint correlation function of parafermionic currents whose non-trivial part is given by  a symmetric homogeneous polynomial vanishing when any $N+1$ particles collide. We refer interested reader to recent papers \cite{bernevig-2008-101,bernevig-2009-102,estienne-2010-824,Estienne:2009mp,Estienne:2010as} for more details and list of all relevant references. In this paper we study these and even more general polynomials which can be identified due to \cite{Estienne:2010as} with three quasihole excitations over the quantum Hall ground state in great details. In particular, we give an explicit representation for them in terms of symmetrization of certain basic ``monomials'' and compute exactly generalized Selberg integral involving these polynomials (see section \ref{integrals}). We also compute the generalization of Mehta integral involving parafermionic polynomial.

Fractional superstring theory was investigated in number of papers \cite{Argyres:1992hc,Argyres:1993wy,Argyres:1993hz,Tye:1993sz,Argyres:1993ar} and recently in \cite{Irie:2009yb,Chan:2010yg} in the context of matrix models. However, many important questions about fractional superstings are still unanswered. One of them is the formulation of non-critical fractional superstring theory within Polyakov's approach similar to  \cite{Polyakov:1981rd,Polyakov:1981re}. This question requires deeper understanding of the nature of parafermionic symmetry, but it is clear that $Z_{N}$ parafermionic analog of Liouville field theory plays here an important role.  Parafermionic Liouville field theory is defined by two-dimensional  Lagrangian density \cite{Fateev:1996ea}
\begin{equation}\label{Lagrangian}
      \mathcal{L}=\mathcal{L}_{\text{PF}}+\frac{N}{4\pi}\left(\partial_a\varphi\right)^2-\mu N\,\Psi\bar{\Psi}e^{2b\varphi},
\end{equation}
where we denote schematically the Lagrangian of $Z_{N}$ parafermionic CFT by $\mathcal{L}_{\text{PF}}$ \cite{Fateev:1985mm} (see also appendix \ref{PF}), $\varphi$ is bosonic field and $\Psi$ is the first parafermionic current. Parameter $\mu$ in \eqref{Lagrangian} is called cosmological constant, while $b$ is coupling constant. Last term in \eqref{Lagrangian} can be viewed as the simplest possible marginal deformation of two theories: free bosonic theory and $Z_{N}$ parafermionic CFT. For $N=1$ the resulting theory coincides with ordinary Liouville field theory \cite{Zamolodchikov:1995aa,Teschner:2001rv} and for $N=2$ with supersymmetric Liouville field theory \cite{Rashkov:1996jx,Poghosian:1996dw,Belavin:2007gz}. For arbitrary values of the parameter $\mu$ the Lagrangian \eqref{Lagrangian}  defines  CFT with central charge
\begin{equation}\label{central-charge}
     c=\frac{3N}{N+2}+\frac{6Q^2}{N},\quad\text{where}\quad Q=b+\frac{1}{b}.
\end{equation}
The symmetry algebra of the theory is actually wider and generated by two holomorphic currents of spin two (stress energy tensor) and non-local current of spin $\frac{N+4}{N+2}$ for $N>1$. For $N=2$ this algebra coincides with Neveu-Shwarz algebra while for $N>2$ it has more complicated structure\footnote{For $N=4$ this algebra can be realized as spin $4/3$ parafermionic algebra \cite{Fateev:1985ig} (see \cite{Argyres:1993hz,Argyres:1993ar}).}. For our purposes we found inconvenient to use standard ``bootstrap'' machinery  based on using of Ward identities an so on \cite{Belavin:1984vu} and prefer to consider the theory \eqref{Lagrangian} as a deformation of the tensor product of two theories free bosonic field theory and $Z_{N}$ parafermionic CFT. It means that we will consider all correlation functions in our model only if the total charge takes special values and they can be represented by finite-dimensional integral of the product of correlation functions in free bosonic theory and parafermionic CFT (see section \ref{3-point} for details). The expression for general values of the total charge can be obtained by means of  ``analytical'' continuation (see section \ref{3-point}). This approach is not mathematically rigorous and the obtained result has to be checked by independent calculations. In section \ref{3-exp-model} we consider the theory \eqref{Lagrangian}  using the dual representation by the three exponential   model which is introduced there and develop a method allowing to study correlation functions without any restrictions on the total charge.

The primary fields in the theory \eqref{Lagrangian} are given by the exponential fields 
\begin{equation}\label{primary}
    V^{(k)}_{\alpha}=\sigma_k e^{2\alpha\varphi},
\end{equation}
where $\sigma_{k}$ are the primary fields in parafermionic CFT  also called order fields (see appendix \ref{PF}). The fields \eqref{primary} are the spinless fields with conformal dimensions
\begin{equation*}
      \Delta^{(k)}(\alpha)=\bar{\Delta}^{(k)}(\alpha)=\frac{k(N-k)}{2N(N+2)}+\frac{\alpha(Q-\alpha)}{N}.
\end{equation*}
In this paper we consider in details the three-point correlation functions of the primary fields \eqref{primary} in parafermionic LFT and obtain generalization of DOZZ formula  \cite{Dorn:1992at,Dorn:1994xn,Zamolodchikov:1995aa} in this case. The paper organized as follows. In section \ref{3-point} we formulate the problem and propose an answer for the three-point correlation function of exponential fields in our model.  We also describe general method used in this paper which leads to the problem of computation of some multi-dimensional integrals analogous to those considered in \cite{Dotsenko:1984ad} but with insertion of certain polynomials (we call them parafermionic polynomials). In section \ref{integrals} (which can be considered as a main novel part of our paper) we compute these integrals using the method adopted from the seminal paper by Selberg \cite{Selberg}. In section \ref{3-exp-model} we consider dual representation of parafermionic LFT given by three exponential model and obtain an analytic expression for the three-point correlation function  which agrees with the result obtained in section \ref{integrals}. Section \ref{3-exp-model} can be considered as an independent check of our results. In this part of the paper we use the notations historically related with the three-exponential model and they  intersect with the notations used in the rest of the text (the same is for appendix \ref{MSP-approximation}). We hope that this fact will not confuse the reader. In section \ref{conclusions} we consider some important questions which were not discussed in the main body of the paper.  In appendix \ref{PF} we give a review of parafermionic CFT \cite{Fateev:1985mm}. In appendices  \ref{conjecture} and \ref{Selberg-prop-proof} we give proofs of two important propositions essentially used in section \ref{integrals}. In appendix \ref{MSP-approximation} we consider minisuperspace approximation for the three-exponential CFT and discuss its sigma-model representation.
\section{Three-point correlation function in parafermionic LFT}\label{3-point}
Our goal is to find the three-point correlation function of the fields \eqref{primary}
\begin{equation}\label{C-def}
    C^{k_1,k_2,k_3}(\alpha_1,\alpha_2,\alpha_3)\overset{\text{def}}{=}
    \langle V^{(k_1)}_{\alpha_1}(0)V^{(k_2)}_{\alpha_2}(1)V^{(k_3)}_{\alpha_3}(\infty)\rangle.
\end{equation}
One can show (see below) that due to the $Z_N$ symmetry of parafermionic CFT the three-point correlation function \eqref{C-def} is non-zero only if one of the following conditions is satisfied\footnote{We note that for even $N$ both conditions \eqref{k-def} can be satisfied at the same time. We discuss this phenomenon below.}
\begin{subequations}\label{k-def}
\begin{align}\label{k-def-1}
     &k_1+k_2+k_3=2k, \quad 0\leq k_j\leq k\leq N,\\
     \intertext{or}\label{k-def-2}
     &k_1+k_2+k_3=N+2k, \quad 0\leq k\leq k_j\leq N.
\end{align}
\end{subequations}
In this paper we give an analytic expression for the three-point correlation function \eqref{C-def}. In order to present an answer it is convenient to define functions
\begin{subequations}\label{both-C}
\begin{multline}\label{C}
    \mathbb{C}^{k_1,k_2,k_3}(\alpha_1,\alpha_2,\alpha_3)=
    \left[\pi\mu\gamma\left(\frac{bQ}{N}\right)b^{-\frac{2bQ}{N}}\right]^{\frac{Q-\alpha}{b}}
    \rho(k_{1},k_{2},k_{3})\times\\\times
    \frac{\Upsilon_{0}^{(N)}(b)\Upsilon_{k_1}^{(N)}(2\alpha_1)\Upsilon_{k_2}^{(N)}(2\alpha_2)\Upsilon_{k_3}^{(N)}(2\alpha_3)}
    {\Upsilon_{k}^{(N)}(\alpha_1+\alpha_2+\alpha_3-Q)\Upsilon_{k-k_1}^{(N)}(\alpha_2+\alpha_3-\alpha_1)
    \Upsilon_{k-k_2}^{(N)}(\alpha_1+\alpha_3-\alpha_2)\Upsilon_{k-k_3}^{(N)}(\alpha_1+\alpha_2-\alpha_3)},
\end{multline}
and 
\begin{multline}\label{C-dual}
    \tilde{\mathbb{C}}^{k_1,k_2,k_3}(\alpha_1,\alpha_2,\alpha_3)=
    \left[\pi\mu\gamma\left(\frac{bQ}{N}\right)b^{-\frac{2bQ}{N}}\right]^{\frac{Q-\alpha}{b}}
    \tilde{\rho}(k_{1},k_{2},k_{3})\times\\\times
    \frac{\Upsilon_{0}^{(N)}(b)\Upsilon_{k_1}^{(N)}(2\alpha_1)\Upsilon_{k_2}^{(N)}(2\alpha_2)\Upsilon_{k_3}^{(N)}(2\alpha_3)}
    {\Upsilon_{k}^{(N)}(\alpha_1+\alpha_2+\alpha_3-Q)\Upsilon_{N+k-k_1}^{(N)}(\alpha_2+\alpha_3-\alpha_1)
    \Upsilon_{N+k-k_2}^{(N)}(\alpha_1+\alpha_3-\alpha_2)\Upsilon_{N+k-k_3}^{(N)}(\alpha_1+\alpha_2-\alpha_3)},
\end{multline}
\end{subequations}
where $\alpha=\alpha_{1}+\alpha_{2}+\alpha_{3}$ and integer parameters $k_{j}$ and $k$ are supposed to be related by \eqref{k-def-1} for the function \eqref{C} and by \eqref{k-def-2} for the function \eqref{C-dual}. The factors $\rho(k_{1},k_{2},k_{3})$ and $\tilde{\rho}(k_{1},k_{2},k_{3})$ in \eqref{C} and \eqref{C-dual}  are functions independent on parameters $\alpha$ and $b$. They  are explicitly  given by
\begin{equation}\label{rho}
  \begin{aligned}
     &\rho(k_{1},k_{2},k_{3})=
     \left[\gamma\left(\frac{1}{N+2}\right)\prod_{j=1}^{3}\gamma\left(\frac{N-k_{j}+1}{N+2}\right)\right]^{\frac{1}{2}}
     \frac{H(N-k)H(k-k_{1})H(k-k_{2})H(k-k_{3})}{H(k_{1})H(k_{2})H(k_{3})},\\
     &\tilde{\rho}(k_{1},k_{2},k_{3})=
     \left[\gamma\left(\frac{1}{N+2}\right)\prod_{j=1}^{3}\gamma\left(\frac{N-k_{j}+1}{N+2}\right)\right]^{\frac{1}{2}}
     \frac{H(k)H(k_{1}-k)H(k_{2}-k)H(k_{3}-k)}{H(k_{1})H(k_{2})H(k_{3})},
  \end{aligned}
\end{equation}
where
\begin{equation*}
      H(k)=\prod_{j=1}^{k}\gamma\left(\frac{j}{N+2}\right).
\end{equation*}
Throughout this paper we use the following standard  notation
\begin{equation}
    \gamma(x)=\frac{\Gamma(x)}{\Gamma(1-x)}.
\end{equation}
In eqs \eqref{both-C} function $\Upsilon_k^{(N)}(x)=\Upsilon_k^{(N)}(x,b)$ is  defined as\footnote{Usually the $b$ dependence of the function $\Upsilon_k^{(N)}(x,b)$ is hidden for shortness.}
\begin{equation}\label{Upsilon-N-def}
        \Upsilon_k^{(N)}(x)=
        \prod_{j=1}^{N-k}\Upsilon\left(\frac{x+kb^{-1}+(j-1)Q}{N}\right)
        \prod_{j=N-k+1}^{N}\Upsilon\left(\frac{x+(k-N)b^{-1}+(j-1)Q}{N}\right),
\end{equation}
where $\Upsilon(x)$ is usual $\Upsilon$-function defined in \cite{Zamolodchikov:1995aa}. It is convenient to supply definition \eqref{Upsilon-N-def} with
\begin{equation*}
     \Upsilon_{N+k}^{(N)}(x)=\Upsilon_k^{(N)}(x),
\end{equation*}
which is self-consistent for $k=0$ as one can readily see from \eqref{Upsilon-N-def}. Function $\Upsilon_k^{(N)}(x)$ obeys several important properties (which follow of course from analogous properties \cite{Zamolodchikov:1995aa} for $\Upsilon$ function).
\begin{subequations}
\paragraph{Reflection:}
\begin{equation}
      \Upsilon_k^{(N)}(Q-x)=\Upsilon_{N-k}^{(N)}(x). 
\end{equation}
\paragraph{Shift:}
\begin{equation}\label{shift}
   \begin{aligned}
       &\Upsilon_k^{(N)}(x+b)=b^{1-\frac{2(bx+k-1)}{N}}\gamma\left(\frac{bx+k-1}{N}\right)\Upsilon_{k-1}^{(N)}(x),\\
       &\Upsilon_k^{(N)}(x+b^{-1})=b^{\frac{2(b^{-1}x+N-k-1)}{N}-1}\gamma\left(\frac{b^{-1}x+N-k-1}{N}\right)
       \Upsilon_{k+1}^{(N)}(x).
   \end{aligned}
\end{equation}
\paragraph{Duality:}
\begin{equation}\label{Upsilon-duality}
        \Upsilon_k^{(N)}(x,b)=\Upsilon_{N-k}^{(N)}(x,b^{-1}).
\end{equation}
\end{subequations}
We note that due to the relation \eqref{shift} all function $\Upsilon_k^{(N)}(x)$ can be obtained from the function $\Upsilon_0^{(N)}(x)$ given by integral representation
\begin{equation}
\log\Upsilon_{0}^{(N)}(x)=\int_{0}^{\infty}
\left(\left(\frac{Q}{2}-x\right)^{2}e^{-2t} -\frac{\sinh Qt\sinh^{2}(\frac{1}{N}\left(\frac{Q}{2}-x\right)t)}{\sinh\frac{Q}{N}t\sinh bt\sinh b^{-1}t}\right)
\frac{dt}{t}.
\end{equation}

In this paper we propose an exact expression for the three-point function \eqref{C-def}. The cases of odd and even $N$ should be considered separately. If $N$ is an odd number then
\begin{equation}\label{C-prop-odd}
   C^{k_1,k_2,k_3}(\alpha_1,\alpha_2,\alpha_3)=
   \begin{cases}
      \mathbb{C}^{k_1,k_2,k_3}(\alpha_1,\alpha_2,\alpha_3)\quad\text{if condition \eqref{k-def-1} is true},\\
      \tilde{\mathbb{C}}^{k_1,k_2,k_3}(\alpha_1,\alpha_2,\alpha_3)\quad\text{if condition \eqref{k-def-2} is true},\\
      0\quad\text{otherwise.}
   \end{cases}
\end{equation}
If $N$ is even then both conditions \eqref{k-def-1} and \eqref{k-def-2} can be satisfied at the same time (of course with different $k$ and $k'$). This is possible only if
$k_{1}+k_{2}+k_{3}=2k\geq N$. In this case our proposal reads as
\begin{multline}\label{C-prop-even}
   C^{k_1,k_2,k_3}(\alpha_1,\alpha_2,\alpha_3)=\\=
   \begin{cases}
      \mathbb{C}^{k_1,k_2,k_3}(\alpha_1,\alpha_2,\alpha_3)\quad\text{if condition \eqref{k-def-1} is true and condition \eqref{k-def-2} is wrong},\\
      \tilde{\mathbb{C}}^{k_1,k_2,k_3}(\alpha_1,\alpha_2,\alpha_3)\quad\text{if condition \eqref{k-def-2} is true and condition \eqref{k-def-1} is wrong },\\
      \mathbb{C}^{k_1,k_2,k_3}(\alpha_1,\alpha_2,\alpha_3)+\tilde{\mathbb{C}}^{k_1,k_2,k_3}(\alpha_1,\alpha_2,\alpha_3)
      \quad\text{if both conditions \eqref{k-def-1} and \eqref{k-def-2} are true},\\
      0\quad\text{otherwise.}
   \end{cases}
\end{multline}
Functions $\mathbb{C}^{k_1,k_2,k_3}(\alpha_1,\alpha_2,\alpha_3)$ and $\tilde{\mathbb{C}}^{k_1,k_2,k_3}(\alpha_1,\alpha_2,\alpha_3)$ in the equations above are defined by \eqref{C} and \eqref{C-dual} respectively. We emphasize that in the definition of $\mathbb{C}^{k_1,k_2,k_3}(\alpha_1,\alpha_2,\alpha_3)$ parameter $k$ is defined by
\begin{equation*}
   k_1+k_2+k_3=2k,
\end{equation*} 
while in the definition of $\tilde{\mathbb{C}}^{k_1,k_2,k_3}(\alpha_1,\alpha_2,\alpha_3)$ by
\begin{equation*}
   k_1+k_2+k_3=N+2k.
\end{equation*} 
In particular, in the third line in the r.h.s in  \eqref{C-prop-even} $\mathbb{C}^{k_1,k_2,k_3}(\alpha_1,\alpha_2,\alpha_3)$ is supposed to depend on parameter $k=(k_{1}+k_{2}+k_{3})/2$ while  $\tilde{\mathbb{C}}^{k_1,k_2,k_3}(\alpha_1,\alpha_2,\alpha_3)$ depends on parameter $k'=k-\frac{N}{2}$. We note that for $N=1$ and $N=2$ our result reproduces known exact expressions for the three-point correlation functions in bosonic Liouville field theory \cite{Zamolodchikov:1995aa,Dorn:1992at,Dorn:1994xn,Teschner:1995yf,Teschner:2003en} and in supersymmetric Liouville field theory \cite{Rashkov:1996jx,Poghosian:1996dw}.  One important comment is needed at this point. Throughout this paper we assume that holomorphic and antiholomorphic parafermionic currents commute with each other which is actually not the case. The locality of the theory requires
\begin{equation}\label{non-locality}
    \Psi(z)\bar{\Psi}(\bar{z}')=\omega\,\bar{\Psi}(\bar{z}')\Psi(z),
\end{equation}
with $\omega=e^{2i\pi/N}$. The commutation rule \eqref{non-locality} gives some phases  which can be easily restored in our final answer \eqref{C-prop-odd} and \eqref{C-prop-even}.  We prefer to hide these phases and put $\omega=1$ for simplicity.

Proposed exact expressions \eqref{C-prop-odd} and \eqref{C-prop-even} posses a symmetry
\begin{equation}\label{theory-duality}
   b\rightarrow\frac{1}{b},\quad
   k_{j}\rightarrow N-k_{j},\quad
   \mu\rightarrow\tilde{\mu},
\end{equation} 
where cosmological constants $\mu$ and $\tilde{\mu}$ are related via
\begin{equation*}
   \left[\pi\mu\gamma\left(\frac{bQ}{N}\right)b^{-\frac{2bQ}{N}}\right]^{\frac{1}{b}}=
   \left[\pi\tilde{\mu}\gamma\left(\frac{b^{-1}Q}{N}\right)b^{\frac{2b^{-1}Q}{N}}\right]^{b}.
\end{equation*}
Transformation \eqref{theory-duality} reflects the duality of parafermionic Liouville theory \eqref{Lagrangian}. Namely, two theories one defined by Lagrangian \eqref{Lagrangian} and another defined by   
\begin{equation}\label{Lagrangian-dual}
      \mathcal{L}=\mathcal{L}_{\text{PF}}+\frac{N}{4\pi}\left(\partial_a\varphi\right)^2-\tilde{\mu}N\,\Psi^{+}\bar{\Psi}^{+}e^{2b^{-1}\varphi},
\end{equation}
are equivalent. Here $\Psi^{+}=\Psi_{N-1}$ is the conjugated parafermionic current (see appendix \ref{PF}). In particular, it follows from this equivalence that parafermionic LFT is unitary at $b=1$.

The quantum field theory defined by the Lagrangian density \eqref{Lagrangian} is interacting field theory for $\mu\neq0$. However, as was noticed by Goulian and Li \cite{Goulian:1990qr} in some cases it can be studied by free-field methods. By definition, the three-point function \eqref{C-def} is given by the path integral\footnote{Equation \eqref{C-fun-int} needs some comments. First of all the integration over $[d\Psi]$ is rather schematic because there is no known Lagrangian description of parafermionic CFT for $N>2$. Second, using projective invariance one can always fix the position of any three conformal fields in arbitrary positions \cite{Belavin:1984vu} (we fix them as usual at $0$, $1$ and $\infty$.)}
\begin{equation}\label{C-fun-int}
    C^{k_1,k_2,k_3}(\alpha_1,\alpha_2,\alpha_3)=
    \int  V^{(k_{1})}_{\alpha_{1}}(0)V^{(k_{2})}_{\alpha_{2}}(1)V^{(k_{3})}_{\alpha_{3}}(\infty)e^{-S[\varphi,\Psi]}[d\varphi][d\Psi]
\end{equation}
One can perform integration over zero mode $\varphi_{0}$ of the field $\varphi=\varphi_{0}+\tilde{\varphi}$ where field  $\tilde{\varphi}$ has no zero Fourier component
\begin{equation}\label{C-fun-int-pole}
   C^{k_1,k_2,k_3}(\alpha_1,\alpha_2,\alpha_3)=\frac{\Gamma(-s)}{2b}
    \int  \tilde{V}^{(k_{1})}_{\alpha_{1}}(0)\tilde{V}^{(k_{2})}_{\alpha_{2}}(1)\tilde{V}^{(k_{3})}_{\alpha_{3}}(\infty)
    \left(-\mu N\int e^{2b\tilde{\varphi}}\Psi\bar{\Psi}d^{2}z\right)^{s}e^{-S_{0}[\tilde{\varphi},\Psi]}[d\tilde{\varphi}][d\Psi],
\end{equation}
where 
\addtocounter{equation}{-1}
\begin{subequations}
\begin{equation}\label{Q-anomaly}
   s=\frac{Q-\alpha}{b},\quad \alpha=\alpha_{1}+\alpha_{2}+\alpha_{3}
\end{equation}
\end{subequations}
and $S_{0}[\tilde{\varphi},\Psi]$ is the action without exponential term  i.e. at $\mu=0$\footnote{The term $Q$ in equation \eqref{Q-anomaly} comes as usual \cite{Zamolodchikov:1995aa} from the curvature term which is missing in \eqref{Lagrangian} for shortness.}. In \eqref{C-fun-int-pole} we use the notation 
$\tilde{V}^{(k)}_{\alpha}=\sigma_{k}e^{2\alpha\tilde{\varphi}}$. According to \eqref{C-fun-int-pole}  the three-point correlation function $C^{k_1,k_2,k_3}(\alpha_1,\alpha_2,\alpha_3)$ has a pole at $s=n$ (i.e. at $\alpha-Q=-nb$) with the residue
\begin{multline}\label{C-fun-int-pole-free-field}
    \text{Res}\Bigl|_{\alpha-Q=-nb}\,C^{k_1,k_2,k_3}(\alpha_1,\alpha_2,\alpha_3)
    =\frac{(\mu N)^{n}}{n!}\times\\\times
    \int\langle\tilde{V}^{(k_{1})}_{\alpha_{1}}(0)\tilde{V}^{(k_{2})}_{\alpha_{2}}(1)\tilde{V}^{(k_{3})}_{\alpha_{3}}(\infty)
     \tilde{V}_{b}^{(0)}(z_{1},\bar{z_{1}})\Psi(z_{1})\bar{\Psi}(\bar{z_{1}})\dots \tilde{V}_{b}^{(0)}(z_{n},\bar{z_{n}})\Psi(z_{n})\bar{\Psi}(\bar{z_{n}})\rangle
     d^{2}z_{1}\dots d^{2}z_{n}
\end{multline}
and $\langle\dots\rangle$ denotes the path integral with the action without exponential term. The integration $\int d^{2}z_{j}$ in \eqref{C-fun-int-pole-free-field} goes over two-dimensional plane. For future convenience we represent parameter $n$ in \eqref{C-fun-int-pole-free-field} as $n=mN-k$, where $m$ and $k$ are arbitrary integers restricted by $m\geq0$ and $0\leq k<N$, i.e. we define the parameter $k$ as a deficit of the parameter $n$ to be divided by $N$. We note that correlation function in the r.h.s. in \eqref{C-fun-int-pole-free-field} essentially factorizes into bosonic and parafermionic parts. The bosonic correlation function in \eqref{C-fun-int-pole-free-field} can be evaluated using Wick rules  with the pairing
\begin{equation*}
  \varphi(z,\bar{z})\varphi(0,0)=-\frac{1}{N}\log|z|+\dots
\end{equation*}
so that \eqref{C-fun-int-pole-free-field} can be rewritten as
\begin{multline}\label{C-2D-Selberg-int}
       \text{Res}\Bigl|_{\alpha-Q=-(mN-k)b}\,C^{k_1,k_2,k_3}(\alpha_1,\alpha_2,\alpha_3)=
       \frac{(\mu N)^{mN-k}}{(mN-k)!}\times\\\times
       \int
       \prod_{i=1}^{mN-k}|z_i|^{2A}|z_i-1|^{2B}\prod_{i<j=1}^{mN-k}|z_i-z_j|^{4g}\;
       \Bigl|\mathcal{P}_{m}^{(N|\,k_1,k_2,k_3)}(z_1,\dots,z_{mN-k})\Bigr|^{2}
       d^{2}z_1\dots d^{2}z_{mN-k}, 
\end{multline}
where 
\begin{equation}\label{ABg-alpha-relation}
   A=-\frac{(2b\alpha_{1}+k_{1})}{N},\quad B=-\frac{(2b\alpha_{2}+k_{2})}{N}\quad\text{and}\quad
   g=-\frac{bQ}{N}.
\end{equation}
Function $\mathcal{P}_{m}^{(N|\,k_1,k_2,k_3)}(z_1,\dots,z_{mN-k})$ \eqref{C-2D-Selberg-int} which is actually a polynomial (see below) is defined through the correlation function in parafermionic CFT as
\begin{multline}\label{P-as-corr}
    \Bigl|\mathcal{P}_{m}^{(N|\,k_1,k_2,k_3)}(z_1,\dots,z_{mN-k})\Bigr|^{2}=
    \prod_{i=1}^{mN-k}|z_i|^{\frac{2k_{1}}{N}}|z_i-1|^{\frac{2k_{2}}{N}}
    \prod_{i<j=1}^{mN-k}|z_i-z_j|^{\frac{4}{N}}\times\\\times
    \langle\sigma_{k_{1}}(0)\sigma_{k_{2}}(1)\sigma_{k_{3}}(\infty)
    \Psi(z_{1})\bar{\Psi}(\bar{z}_{1})\dots \Psi(z_{mN-k})\bar{\Psi}(\bar{z}_{mN-k})\rangle.
\end{multline}
One can easily show (see appendix \ref{PF}) that this correlation function is non-zero only if the parameters $k$ and $k_{j}$ are related via one of the conditions \eqref{k-def}. Below we will consider in details the case  when the parameters $k$ and $k_{j}$ are related as in \eqref{k-def-1} with the case  of the relation \eqref{k-def-2} being done in complete analogy (we present some relevant formulae for this case in the end of appendix \ref{PF}). We also denote throughout this paper
\begin{equation}\label{n}
  n=mN-k,
\end{equation}
so in any expressions below parameter $n$ is always given by \eqref{n} (except section \ref{3-exp-model} and appendix \ref{MSP-approximation} which have their own notations). Under above assumptions we left with the problem of computation of the Coulomb integral
\begin{multline}\label{2D-Selberg-int}
       \mathcal{G}_{m}^{(N|\,k_1,k_2,k_3)}(A,B|g)=\frac{(\mu N)^{n}}{n!}\times\\\times
       \int
       \prod_{i=1}^{n}|z_i|^{2A}|z_i-1|^{2B}\prod_{i<j=1}^{n}|z_i-z_j|^{4g}\;
       \Bigl| \mathcal{P}_{m}^{(N|\,k_1,k_2,k_3)}(z_1,\dots,z_{n})\Bigr|^{2}
       d^{2}z_1\dots d^{2}z_{n}, 
\end{multline}
which can be viewed as a generalization of the similar Coulomb integrals studied in \cite{Dotsenko:1984ad, Dotsenko:1984nm} in the case of bosonic theory ($N=1$ in our case) and in \cite{AlvarezGaume:1991bj} in the case of supersymmetric theory ($N=2$).   It is clear that in some domain of the parameters $A$, $B$ and $g$ this integral is convergent. If one succeeded  and computed the integral \eqref{2D-Selberg-int} (which is usually given in terms of product of $\Gamma$-functions) then one can formally continue it to arbitrary values of the parameter $s=(Q-\alpha)/b$ (see below). This method was used in number of papers \cite{Zamolodchikov:1995aa,Rashkov:1996jx,Fateev:2007qn,Fateev:2007ab,Fateev:2008bm,Fateev:2007tt} and in all cases it gives correct expression for the three-point correlation function (it is usually confirmed by other methods like semiclassical approximation, bootstrap approach etc).
\section{Parafermionic polynomials and Selberg integrals}\label{integrals}
The main subject of this section is to compute the integral \eqref{2D-Selberg-int}. As a warm-up exercise we consider contour version of the integral  \eqref{2D-Selberg-int}
\begin{equation}\label{Selberg-integral}
     \mathbf{J}_{m}^{(N|\,k_1,k_2,k_3)}(A,B|g)=\int\limits_0^1\dots\int\limits_0^1
       \prod_{i=1}^{n}t_i^{A}(1-t_i)^{B}\prod_{i<j=1}^{n}|t_i-t_j|^{2g}\;
       \mathbf{P}_{m}^{(N|\,k_1,k_2,k_3)}(t_1,\dots,t_{n})
       dt_1\dots dt_{n}, 
\end{equation}
where integers $k_1$, $k_2$, $k_3$ and $k$ are restricted by the condition
\begin{equation}\label{cond}
    k_1+k_2+k_3=2k, \quad 0\leq k_j\leq k\leq N,
\end{equation}
and $n=mN-k$. Polynomial $\mathbf{P}_{m}^{(N|\,k_1,k_2,k_3)}(t_1,\dots,t_{n})$ is symmetric polynomial defined through the more general polynomial $\mathbf{P}_{m}^{(N|\,k_1,k_2,k_3)}(x_1,x_2,x_3|t_1,\dots,t_{n})$ as
\begin{equation}\label{pol-pol-def}
   \mathbf{P}_{m}^{(N|\,k_1,k_2,k_3)}(t_1,\dots,t_{n})=(-1)^{(m-1)k_1}
   \lim_{x_3\rightarrow\infty}x_3^{-(m-1)k_3}\,\mathbf{P}_{m}^{(N|\,k_1,k_2,k_3)}(0,1,x_3|t_1,\dots,t_{n}).
\end{equation}
We stress that the polynomial $\mathbf{P}_{m}^{(N|\,k_1,k_2,k_3)}(t_1,\dots,t_{n})$ differs from  the polynomial $\mathcal{P}_{m}^{(N|\,k_1,k_2,k_3)}(t_1,\dots,t_{n})$ defined by \eqref{P-as-corr} only by normalization
\begin{equation}\label{P-PP}
     \mathcal{P}_{m}^{(N|\,k_1,k_2,k_3)}(t_1,\dots,t_{n})=
         (-1)^{(m-1)k_{1}}\left[\frac{k!(N-k)!}{N!}\,\rho(k_{1},k_{2},k_{3})\right]^{\frac{1}{2}}\,
         \mathbf{P}_{m}^{(N|\,k_1,k_2,k_3)}(t_1,\dots,t_{n})
\end{equation}
where $\rho(k_{1},k_{2},k_{3})$ can be found in \eqref{rho}. Our choice of normalization will become clear below (see appendix \ref{PF} and especially eq \eqref{P-normalization}).

Polynomial $\mathbf{P}_{m}^{(N|\,k_1,k_2,k_3)}(x_1,x_2,x_3|t_1,\dots,t_{n})$ is symmetric polynomial in variables $t_1,\dots,t_{n}$ (and not necessary symmetric in $x_{j}$) satisfying the set of properties which follow from the definition \eqref{P-as-corr} (see appendix \ref{PF})\footnote{In order to save space we will denote $\mathbf{P}_{m}^{(N|\,k_1,k_2,k_3)}(x_1,x_2,x_3|t_1,\dots,t_{n})$ simply as $\mathbf{P}(x_1,x_2,x_3|t_1,\dots,t_{n})$ or even $\mathbf{P}(x_i|t_j)$ when needed.}
\begin{subequations}\label{prop}
   \begin{align}\label{prop-1} 
       &\mathbf{P}(\lambda x_i|\lambda t_{j})=
        \lambda^{m(m-1)N}\mathbf{P}(x_i|t_{j}),\\\label{prop-2}
       &\mathbf{P}\left(x_i^{-1}\bigl|t_j^{-1}\right)=
        \prod_{j=1}^3x_j^{-(m-1)k_j}\prod_{q=1}^{n}t_q^{-2(m-1)}\mathbf{P}(x_i|t_j),\\\label{prop-3}
       &\mathbf{P}(x_i+\lambda|t_j+\lambda)=\mathbf{P}(x_i|t_j),\\[2mm]\label{prop-4}
       &\mathbf{P}(x_1,x_2,x_3|\underbrace{x,\dots x}_{N+1},t_1,\dots t_{n-N-1})=0,\\\label{prop-5}
       &\mathbf{P}(x_1,x_2,x_3|\underbrace{x_j,\dots x_j}_{N+1-k_j},t_1,\dots t_{n+k_j-N-1})=0\quad\text{for}\;j=1,2,3.
   \end{align}
\end{subequations}
\begin{proposition}\label{main-proposition}
Polynomial  $\mathbf{P}_{m}^{(N|\,k_1,k_2,k_3)}(x_1,x_2,x_3|t_1,\dots,t_{n})$ is uniquely determined by the set of properties \eqref{prop} up to a total factor.
\end{proposition}
The proof of proposition \ref{main-proposition} is given  in appendix \ref{conjecture}.  The polynomial $\mathbf{P}_{m}^{(N|\,k_1,k_2,k_3)}(x_1,x_2,x_3|t_1,\dots,t_{n})$ can be explicitly written in terms of symmetrization of certain basic "monomial" $p(x_1,x_2,x_3|t_1,\dots,t_{n})$ as
\begin{equation}\label{Polynomial-definition}
     \mathbf{P}_{m}^{(N|\,k_1,k_2,k_3)}(x_1,x_2,x_3|t_1,\dots,t_{n})=
      \Lambda_{m}^{-1}\,S[p(x_1,x_2,x_3|t_1,\dots,t_{n})],
\end{equation}
where $S$ means symmetrization over all variables $(t_1,\dots,t_{n})$. Normalization constant
\begin{equation}\label{Lambda}
    \Lambda_{m}=\left[\frac{N^{n}k!(N-k)!}{N!}\right]^{\frac{1}{2}}
    (m!)^{N-k}((m-1)!)^{k}
\end{equation}
in \eqref{Polynomial-definition} were chosen in such a way that the  polynomial explicitly given by \eqref{Polynomial-definition} together with \eqref{P-PP} gives  \eqref{P-as-corr}.

The basic ``monomial'' $p(x_1,x_2,x_3|t_1,\dots,t_n)$ in \eqref{Polynomial-definition} can be defined as follows. We divide our points $t_1,\dots,t_n$ into $N$ groups
\begin{equation*}
    (t_1,t_2,\dots,t_n)\longrightarrow(s_1,\dots,s_N),
\end{equation*}
where  $N-k$ groups  contain $m$ different elements and the remaining $k$ groups  contain $m-1$ different elements. We associate with each group $s_q$ the "monomial" $\prod_{i<j\in s_q}(t_i-t_j)^2$ and take the product over all groups
\begin{equation}\label{group-prod-monomial}
        \prod_{q=1}^N\prod_{i<j\in s_q}(t_i-t_j)^2.
\end{equation}
It is clear that if the basic "monomial" is proportional to the product \eqref{group-prod-monomial} then the clustering property \eqref{prop-4}  follows automatically because if one sets any $N+1$ variables $t_{j}$ to the same value then at least two of them will be inside the same group and hence the expression \eqref{group-prod-monomial} vanish. The rest of the basic monomial  is constructed as follows. All groups $(s_1,\dots,s_N)$ are divided into four sets with $k-k_1$, $k-k_2$ and $k-k_3$ groups ($(m-1)$ elements in each group) and $N-k$ groups ($m$ elements in each group)
\begin{equation*}
  \begin{aligned}
     &\chi_1=(s_{1},\dots,s_{k-k_1}),\quad&
     &\chi_2=(s_{k-k_1+1},\dots,s_{k_3}),\\
     &\chi_3=(s_{k_3+1},\dots,s_{k}),\quad&
     &\chi_4=(s_{k+1},\dots,s_{N}).
  \end{aligned}
\end{equation*} 
We define the basic monomial as
\begin{multline}\label{basic-monomial}
     p(x_1,x_2,x_3|t_1,\dots,t_{n})=\\=
     \prod_{q=1}^N\prod_{i<j\in s_q}(t_i-t_j)^2\,
     \prod_{j\in\chi_1}(t_j-x_2)(t_j-x_3)\prod_{j\in\chi_2}(t_j-x_1)(t_j-x_3)\prod_{j\in\chi_3}(t_j-x_1)(t_j-x_2).
\end{multline}
We note that the expression involving variable $x_j$ in \eqref{basic-monomial} looks like
\begin{equation}\label{x-monomial}
      \prod_{i\in s_{p_1}}(t_i-x_j)\dots\prod_{i\in s_{p_{k_j}}}(t_i-x_j),
\end{equation}
i.~e. it involves the product over exactly $k_j$ groups for any $j=1,2,3$ and apparently ensures clustering property \eqref{prop-5}. After symmetrization over $(t_{1},\dots,t_{n})$ we get the polynomial which is non-zero and it it is easy to see that  the rest of properties in \eqref{prop} are also valid. With the polynomial defined by \eqref{Polynomial-definition} and \eqref{basic-monomial} we define the polynomial $\mathbf{P}_{m}^{(N|\,k_1,k_2,k_3)}(t_1,\dots,t_{n})$ as in \eqref{pol-pol-def}. From the standard CFT point of view the definition \eqref{pol-pol-def} is  equivalent to using the projective invariance of correlation function 
\begin{equation*}
  \langle\sigma_{k_{1}}(x_{1})\sigma_{k_{2}}(x_{2})\sigma_{k_{3}}(x_{3})\Psi(t_{1})\dots\Psi(t_{n})\rangle\rightarrow
  \langle\sigma_{k_{1}}(0)\sigma_{k_{2}}(1)\sigma_{k_{3}}(\infty)\Psi(t_{1})\dots\Psi(t_{n})\rangle,
\end{equation*}
allowing to fix $x_{1}=0$, $x_{2}=1$ and $x_{3}=\infty$.

Integrals of the form \eqref{Selberg-integral} i.e.
\begin{equation}\label{Selberg-sample}
   \int\limits_0^1\dots\int\limits_0^1
       \prod_{j=1}^{n}t_j^{A}(1-t_j)^{B}\prod_{i<j=1}^{n}|t_i-t_j|^{2g}\;
       P(t_1,\dots,t_{n})\,dt_{1}\dots dt_{n},
\end{equation}
where $P(t_1,\dots,t_{n})$ is some symmetric polynomial were extensively studied in the literature. When $P(t_1,\dots,t_{n})=1$ the integral \eqref{Selberg-sample} coincides with ordinary Selberg integral \cite{Selberg}.  The case when $P(t_1,\dots,t_{n})$ is elementary symmetric polynomial was considered by Aomoto \cite{MR876291}. Further generalization has been computed by Kadell \cite{MR1467311}. In this case  $P(t_1,\dots,t_{n})$ coincides with Jack polynomial $\jac_{\lambda}^{1/g}(t_{1},\dots,t_{n})$, where $\lambda$ is arbitrary partition of at most $n$ parts (for further development see also \cite{Iguri:2009fk,Iguri:2009uq}). In this paper we present another exact results for the integrals of the form \eqref{Selberg-sample} (see below).

We will follow mainly original approach of Selberg \cite{Selberg} also described in \cite{Forrester}. For $g$ being \emph{non-negative integer} number we can expand in \eqref{Selberg-integral}
\begin{equation}\label{basic-pol-as-list-of-monomials}
       \prod_{i<j=1}^{n}|t_i-t_j|^{2g}\;\mathbf{P}_{m}^{(N|\,k_1,k_2,k_3)}(t_1,\dots,t_{n})=
       \sum_{\nu_j}C_{\nu_1,\dots,\nu_{n}}\;t_1^{\nu_1}\dots t_{n}^{\nu_{n}},
\end{equation}
where the sum goes over all eligible sets of integers $\nu_{1},\dots,\nu_{n}$.
It is clear that the values of the coefficients $C_{\nu_1,\dots,\nu_{n}}$ in \eqref{basic-pol-as-list-of-monomials} are independent on the order of the integers $\nu_1,\dots,\nu_{n}$. It happens that in complete analogy with Selberg case (i.e. in the case with polynomial being a constant) the integers arranged as $\nu_1\leq \nu_2\dots\leq \nu_{n}$ satisfy proper set of inequalities which allow to find $A$ and $B$ dependent part of the integral \eqref{Selberg-integral} analytically.
\begin{proposition}\label{Selberg-prop-1}
For $\nu_1\leq \nu_2\dots\leq \nu_{n}$ the non-zero values of $C_{\nu_1,\dots,\nu_{n}}$ in \eqref{basic-pol-as-list-of-monomials} occur when the integers $\nu_{pN+j}$  satisfy the inequalities
\begin{subequations}\label{ineq}
\begin{align}\label{ineq1}
  &\begin{aligned}
        &\nu_{pN+j}\geq p+(pN+j-1)g&&\;\text{if}\quad p=0,\dots,m-2,\quad j=1,\dots,N-k_1,\\
        &\nu_{pN+j}\geq p+1+(pN+j-1)g&&\;\text{if}\quad p=0,\dots,m-2,\quad j=N-k_1+1,\dots,N,\\
        &\nu_{N(m-1)+j}\geq m-1+((m-1)N+j-1)g&&\;\text{if}\quad j=1,\dots, N-k,
  \end{aligned}\\[5mm]\label{ineq2}
  &\begin{aligned}
        &\nu_{pN+j}\leq m+p-1+((m+p)N-k+j-2)g&&\;\text{if}\quad p=0,\dots,m-1,\; j=1,\dots,N-k,\\
        &\nu_{pN+j}\leq m+p-1+((m+p)N-k+j-2)g&&\;\text{if}\quad p=0,\dots,m-2,\; j=N-k+1,\dots,N-k+k_3,\\
        &\nu_{pN+j}\leq m+p+((m+p)N-k+j-2)g&&\;\text{if}\quad p=0,\dots,m-2,\; j=N-k+k_3+1,\dots,N.
   \end{aligned}
\end{align}
\end{subequations}
\end{proposition}
We give a proof of proposition \ref{Selberg-prop-1} in appendix \ref{Selberg-prop-proof}. Substituting expansion \eqref{basic-pol-as-list-of-monomials} in the definition \eqref{Selberg-integral} of the integral $\mathbf{J}_{m}^{(N|\,k_1,k_2,k_3)}(A,B|g)$ and using the Euler beta integral formula we get
\begin{equation}\label{Int-rep}
      \mathbf{J}_{m}^{(N|\,k_1,k_2,k_3)}(A,B|g)=\sum_{\nu_{j}}C_{\nu_1,\dots,\nu_{n}}
      \prod_{j=1}^{n}\frac{\Gamma(1+A+\nu_j)\Gamma(1+B)}{\Gamma(2+A+B+\nu_j)}.
\end{equation}
Using the first set of inequalities \eqref{ineq1} in the numerator in \eqref{Int-rep} the second set  \eqref{ineq2} in the denominator and the relation $\Gamma(1+x)=x\,\Gamma(x)$ we can rewrite
\begin{multline}\label{Int-rep-1}
      \prod_{j=1}^{n}\frac{\Gamma(1+A+\nu_j)}{\Gamma(2+A+B+\nu_j)}=H_{\{\nu_j\}}(A,B)
      \prod_{p=0}^{m-2}\frac{G_{k_1}^{(N)}(1+A+pNg+p)}{G_{k-k_3}^{(N)}\left(1+A+B+((m+p)N-k-1)g+m+p\right)}
      \times\\\times
      \frac{G_0^{(N-k)}(m+A+(m-1)Ng)}{G_0^{(N-k)}(2m+A+B+((2m-1)N-k-1)g)},
\end{multline}
where
\begin{equation}\label{G-def}
    G_{k}^{(N)}(x)\overset{\text{def}}{=}
       \prod_{j=1}^{N-k}\Gamma(x+(j-1)g)\prod_{j=N-k+1}^{N}\Gamma(1+x+(j-1)g),
\end{equation}
and $H_{\{\nu_j\}}(A,B)$ is some polynomial in $A$, $B$ and $g$ (which depends on the integers $\{\nu_j\}=(\nu_1,\dots,\nu_{n})$). By symmetry reasons it is also useful to represent in \eqref{Int-rep}
\begin{equation}\label{Int-rep-2}
     \prod_{j=1}^{n}\Gamma(1+B)=\frac{1}{Q_{k_2}(B)}\prod_{p=0}^{m-2} G_{k_2}^{(N)}(1+B+pNg+p)\;
      G_0^{(N-k)}(m+B+(m-1)Ng),
\end{equation}
with some polynomial $Q_{k_2}(B)$ (it does not depend on choice of integers $\{\nu_j\}$, but its explicit form is not important for us now). Using \eqref{Int-rep-1} and \eqref{Int-rep-2} we get
\begin{multline}\label{Int-A-rep}
    \mathbf{J}_{m}^{(N|\,k_1,k_2,k_3)}(A,B|g)=\frac{H(A,B)}{Q_{k_2}(B)}\;
    \prod_{p=0}^{m-2}\frac{G_{k_1}^{(N)}(1+A+pNg+p)\,G_{k_2}^{(N)}(1+B+pNg+p)}
     {G_{k-k_3}^{(N)}\left(1+A+B+((m+p)N-k-1)g+m+p\right)}
      \times\\\times
      \frac{G_0^{(N-k)}(m+A+(m-1)Ng)\,G_0^{(N-k)}(m+B+(m-1)Ng)}{G_0^{(N-k)}(2m+A+B+((2m-1)N-k-1)g)},
\end{multline}
where
\begin{equation*}
        H(A,B)\overset{\text{def}}{=}\;\sum_{\nu_{j}}C_{\nu_1,\dots,\nu_{n}}\,H_{\{\nu_j\}}(A,B).
\end{equation*}
Now due to the symmetry of the integral \eqref{Selberg-integral} with respect to $A\leftrightarrow B$ and $k_1\leftrightarrow k_2$ we can also represent it as
\begin{multline}\label{Int-B-rep}
    \mathbf{J}_{m}^{(N|\,k_1,k_2,k_3)}(A,B|g)=\frac{\tilde{H}(A,B)}{Q_{k_1}(A)}\;
    \prod_{p=0}^{m-2}\frac{G_{k_1}^{(N)}(1+A+pNg+p)\,G_{k_2}^{(N)}(1+B+pNg+p)}
     {G_{k-k_3}^{(N)}\left(1+A+B+((m+p)N-k-1)g+m+p\right)}
      \times\\\times
      \frac{G_0^{(N-k)}(m+A+(m-1)Ng)\,G_0^{(N-k)}(m+B+(m-1)Ng)}{G_0^{(N-k)}(2m+A+B+((2m-1)N-k-1)g)},
\end{multline}
with some another polynomials $\tilde{H}(A,B)$ and $Q_{k_1}(A)$. The validity of both equations \eqref{Int-A-rep} and \eqref{Int-B-rep} requires
\begin{equation*}
    \frac{H(A,B)}{Q_{k_2}(B)}=\frac{\tilde{H}(A,B)}{Q_{k_1}(A)},
\end{equation*}
which means in particular that $H(A,B)$ is completely divisible by $Q_{k_2}(B)$ as well as $\tilde{H}(A,B)$ is completely divisible by $Q_{k_1}(A)$
\begin{equation*}
     \frac{H(A,B)}{Q_{k_2}(B)}=\frac{\tilde{H}(A,B)}{Q_{k_1}(A)}=R(A,B),
\end{equation*}
where $R(A,B)$ is some polynomial in $A$ and $B$  of certain degree\footnote{It may depend on parameter $g$ in non-polynomial way, which is actually the case (see below).}. Let us show, that in fact $R(A,B)$ does not depend on parameters $A$ and $B$. First, we consider large $A$ asymptotic of the integral \eqref{Selberg-integral}. It can be done by changing variables $t_k=e^{-\frac{s_k}{A}}$ in \eqref{Selberg-integral}
\begin{equation}\label{deg-1}
       \mathbf{J}_{m}^{(N|\,k_1,k_2,k_3)}(A,B|g)\sim A^{-n(m+B+(n-1)g)+(m-1)(k-k_2)}. 
\end{equation}
From other side using
\begin{equation*}
     \frac{\Gamma(x+\alpha)}{\Gamma(x)}\sim x^{\alpha},
     \quad\text{as}\quad x\rightarrow\infty,\quad\arg(x)\neq\pi,
\end{equation*}
in \eqref{Int-A-rep} we get
\begin{multline}\label{deg-2}
    \prod_{p=0}^{m-2}\frac{G_{k_1}^{(N)}(1+A+pNg+p)}{G_{k-k_3}^{(N)}\left(1+A+B+((m+p)N-k-1)g+m+p\right)}
    \times\\\times
    \frac{G_0^{(N-k)}(m+A+(m-1)Ng)}{G_0^{(N-k)}(2m+A+B+((2m-1)N-k-1)g)}
    \sim A^{-n(m+B+(n-1)g)+(m-1)(k-k_2)}.
\end{multline}
Because both \eqref{deg-1} and \eqref{deg-2} have the same behavior at $A\rightarrow\infty$ we conclude that the polynomial $R(A,B)$ does not depend on $A$. Analogously we find that it does not depend on $B$ as well. So $R(A,B)$ may depend only on parameter $g$, i.e.
\begin{multline}\label{Int-g-rep}
    \mathbf{J}_{m}^{(N|\,k_1,k_2,k_3)}(A,B|g)=C_{m}^{(N,k)}(g)\;
    \prod_{p=0}^{m-2}\frac{G_{k_1}^{(N)}(1+A+pNg+p)\,G_{k_2}^{(N)}(1+B+pNg+p)}
     {G_{k-k_3}^{(N)}\left(1+A+B+((m+p)N-k-1)g+m+p\right)}
      \times\\\times
      \frac{G_0^{(N-k)}(m+A+(m-1)Ng)\,G_0^{(N-k)}(m+B+(m-1)Ng)}{G_0^{(N-k)}(2m+A+B+((2m-1)N-k-1)g)},
\end{multline}
The  factor $C_{m}^{(N,k)}(g)$ in \eqref{Int-g-rep} depends on parameters $k_{1}$, $k_{2}$ and $k_{3}$ only through their sum i.e.  it depends only on parameter $k$. This fact can be easily seen if we take
\begin{equation}\label{Selberg-Mehta-cond}
     t_k=\frac{1}{2}-\frac{s_k}{2L},\qquad A=B=\frac{L^2}{2},
\end{equation}
consider the limit $L\rightarrow\infty$ in the integral \eqref{Selberg-integral}\footnote{While performing the limit of the integral \eqref{Selberg-integral} we used the following limit
\begin{equation*}
   \lim_{L\rightarrow\infty}\left(1-\left(\frac{x}{L}\right)^{2}\right)^{\frac{L^{2}}{2}}=e^{-\frac{x^{2}}{2}}.
\end{equation*}}
and from other side use Stirling formula in the r.h.s. in \eqref{Int-g-rep}. One obtains that  $C_{m}^{(N,k)}(g)$ is given by generalized Mehta integral
\begin{equation}\label{Mehta-integral}
    C_{m}^{(N,k)}(g)=\left(\frac{1}{\sqrt{2\pi}}\right)^{n}
    \int\limits_{-\infty}^{\infty}\dots\int\limits_{-\infty}^{\infty}\;
    \prod_{j=1}^{n}e^{-\frac{s_j^2}{2}}\;\prod_{i<j}|s_i-s_j|^{2g}
    \;\mathfrak{P}_{m}^{(N)}(s_1,\dots,s_{n})\;
    ds_1\dots ds_{n},
\end{equation}
here $\mathfrak{P}_{m}^{(N)}(s_1,\dots,s_{n})$ is the polynomial given by  symmetrization of the "monomial" $\mathfrak{p}(s_1,\dots,s_{n})$
\begin{equation}\label{Pol-sim-to-Jack}
   \mathfrak{P}_{m}^{(N)}(s_1,\dots,s_{n})=\Lambda_{m}^{-1}\,S\left[\mathfrak{p}(s_1,\dots,s_{n})\right],
\end{equation}
where normalization factor $\Lambda_{m}$ is given by \eqref{Lambda} and 
\begin{equation*}
   \mathfrak{p}(s_1,\dots,s_{n})=
   \prod_{p=0}^{N-k-1}\prod_{i<j=1}^m\left(s_{i+mp}-s_{j+mp}\right)^2
   \prod_{p=N-k}^{N-1}\prod_{i<j=1}^{m-1}\left(s_{i+(m-1)p}-s_{j+(m-1)p}\right)^2.
\end{equation*}
We see that \eqref{Mehta-integral} indeed depends only on parameter $k$ and not on parameters $k_{j}$.
Polynomial \eqref{Pol-sim-to-Jack} is symmetric homogeneous polynomial of total degree $(m-1)(n-k)$ (we remind that $n=mN-k$) satisfying clustering property 
\begin{equation}\label{Jack-cluster}
      \mathfrak{P}_{m}^{(N)}(\underbrace{x,\dots,x}_{N+1},s_1,\dots,s_{n-N-1})=0.
\end{equation}
We note that the polynomial \eqref{Pol-sim-to-Jack} can be obtained from the polynomial  $\mathbf{P}_{m}^{(N|\,k_{1},k_{2},k_{3})}(t_{1},\dots,t_{n})$ with variables $t_{j}$ taken as in \eqref{Selberg-Mehta-cond} in the limit $L\rightarrow\infty$.  As was shown in \cite{Feigin:1993qr,feigin-2001} the degree $(m-1)(n-k)$ is the minimal possible degree for the symmetric homogeneous polynomial in $n=mN-k$ variables satisfying \eqref{Jack-cluster} and this unique polynomial is proportional to Jack polynomial \cite{Macdonald}
\begin{equation}\label{proportional-to-Jack}
   \mathfrak{P}_{m}^{(N)}(s_1,\dots,s_{n})\sim\jac^{\alpha_N}_{\lambda}(s_1,\dots,s_{n}),
\end{equation}
where $\alpha_N=-N-1$ and $\lambda$ is the special partition
\begin{equation*}\label{partition}
   \lambda=\Bigl[\,\underbrace{2(m-1),\dots,2(m-1)}_{N-k},\,\underbrace{2(m-2),\dots,2(m-2)}_{N}
            ,\dots,\,\underbrace{2,\dots,2}_{N}\,\Bigr].
\end{equation*}
Since the factor $C_{m}^{(N,k)}(g)$ does not depend on parameters $k_{j}$ we found convenient to evaluate it from the integral  \eqref{Selberg-integral} in the case $k_{1}=0$ and $k_{2}=k_{3}=k$. By Carlson's theorem \cite{Carlson}\footnote{The only condition we need is that r.h.s. of \eqref{Int-g-rep} grows as $O(e^{\mu|g|})$ with $\mu<\pi$. This assumption will be justified by explicit form of $C_{m}^{(N,k)}(g)$.} we find that \eqref{Int-g-rep} holds not necessary for integer values of the parameter $g$. In particular it can be continued to the region of sufficiently small negative $g$. In this region function $C_{m}^{(N)}(g)$ can be found if we change variables
\begin{equation*}
    t_j=\tau\xi_j\quad\text{with}\quad\sum\xi_j=1\quad\text{for}\quad j=1\dots N,
\end{equation*}
in the integral \eqref{Selberg-integral} with $k_{1}=0$ and $k_{2}=k_{3}=k$
\begin{multline}\label{Int-limit}
       \int\limits_0^Nd\tau\,\tau^{NA+N(N-1)g+N-1}
       \int\limits_0^1d\xi_1\dots\int\limits_0^1d\xi_N
       \prod_{j=1}^N\xi_j^A(1-\xi_j\tau)^B\prod_{i<j}|\xi_i-\xi_j|^{2g}\delta(\xi_1+\dots+\xi_N-1)\cdot\\
       \int\limits_0^1\dots\int\limits_0^1
       \prod_{i=1}^{n-N}t_i^{A}(1-t_i)^{B}\prod_{j=1}^N|t_i-\tau\xi_j|^{2g}
       \prod_{i<j=1}^{n-N}|t_i-t_j|^{2g}\;
       \mathbf{P}_{m}^{(N|0,k,k)}(\tau\xi_1,\dots\tau\xi_N,t_1,\dots,t_{n-N})
       dt_{1}\dots dt_{n}.
\end{multline}
The integral over $d\tau$ in \eqref{Int-limit} is divergent logarithmically at $A=-1-(N-1)g$. Applying the general formula
\begin{equation*}
    \lim_{a\rightarrow(-1)^{+}}(1+a)\int_{0}^{\Lambda}x^{a}f(x)\,dx=f(0),
\end{equation*}
for any function $f(x)$ continuous at $x\rightarrow0^{+}$ we can consider the limit 
\begin{equation*}
       \lim_{A\rightarrow-1-(N-1)g}(1+A+(N-1)g)\,\eqref{Int-limit}.
\end{equation*}
Using the clustering property
\begin{equation*}
    \mathbf{P}_{m}^{(N|0,k,k)}(\underbrace{x,\dots x}_{N},t_1,\dots,t_{n-N})=\frac{N!}{N^{\frac{N}{2}}}
    \prod_{j=1}^{n-N}(t_{j}-x)^{2}\,
    \mathbf{P}_{m-1}^{(N|0,k,k)}(t_1,\dots,t_{n-N}),
\end{equation*}
which follows from the explicit form of the polynomial $\mathbf{P}_{m}^{(N|0,k,k)}(t_1,\dots,t_{n})$ one immediately obtains
\begin{equation}\label{JJ-relat}
       \lim_{A\rightarrow-1-(N-1)g}(1+A+(N-1)g)\,\mathbf{J}_{m}^{(N|0,k,k)}(A,B|g)=
       N^{-\frac{N}{2}}\,\frac{n!}{(n-N)!}\;X_N(g)\,\mathbf{J}_{m-1}^{(N|0,k,k)}(1+(N+1)g,B|g),  
\end{equation}
where
\begin{subequations}
\begin{equation*}\label{XN}
     X_N(g)=\frac{1}{N}\int\limits_0^1\dots\int\limits_0^1
     \prod_{k=1}^N\xi_k^{-1-(N-1)g}\prod_{i<j}|\xi_i-\xi_j|^{2g}\delta(\xi_1+\dots+\xi_N-1)\,d\xi_1\dots d\xi_N.
\end{equation*}
\end{subequations}
In order to compute $X_{N}(g)$ we set in \eqref{JJ-relat} $m=1$ and $k=0$. In this case the polynomial entering the integral $\mathbf{J}_{1}^{(N|0,0,0)}(A,B|g)$ is actually some constant so that  $\mathbf{J}_{1}^{(N|0,0,0)}(A,B|g)$ is ordinary Selberg integral which is already known. Considering the limit \eqref{JJ-relat} in this case we get
\begin{equation}
     X_N(g)=
     \frac{1}{\Gamma(-Ng)}\prod_{j=1}^N\frac{\Gamma(-jg)\Gamma(1+jg)}{\Gamma(1+g)}.
\end{equation}
Now substituting \eqref{Int-g-rep} in  \eqref{JJ-relat} we find that
\begin{equation}\label{C-k-reccursion}
     C_m^{(N,\,k)}(g)=N^{-\frac{N}{2}}\,\frac{n!}{(n-N)!}\;\frac{G_{N-k}^{(N)}(m-1+((m-1)N-k+1)g)}{\Gamma^N(1+g)}\,C_{m-1}^{(N,\,k)}(g).
\end{equation}
Repeating \eqref{C-k-reccursion} $(m-2)$ times we get
\begin{equation}\label{C-k-answer-prem}
      C_m^{(N,\,k)}(g)=N^{-\frac{(m-1)N}{2}}\,\frac{n!}{(N-k)!}
      \prod_{p=1}^{m-1}\;\frac{G_{N-k}^{(N)}(p+(pN-k+1)g)}{\Gamma^N(1+g)}\,C_1^{(N,\,k)}(g).
\end{equation}
Finally using that under chosen normalization
\begin{equation*}
    \mathfrak{P}_{1}^{(N)}(t_1,\dots,t_{N-k})=\left(\frac{N!(N-k)!}{N^{N-k}\,k!}\right)^{\frac{1}{2}}    
\end{equation*}
and hence the integral \eqref{Mehta-integral} is just the ordinary Mehta integral in this case we find that
\begin{equation}
       C_1^{(N,\,k)}(g)=\left(\frac{N!(N-k)!}{N^{N-k}\,k!}\right)^{\frac{1}{2}}\prod_{j=1}^{N-k}\frac{\Gamma(1+jg)}{\Gamma(1+g)}.
\end{equation}
Combining altogether one arrives at
\begin{equation}\label{C-k-answer}
      C_m^{(N,\,k)}(g)=n!\left[N^{-n}\frac{N!}{k!(N-k)!}\right]^{\frac{1}{2}} 
      \left(\frac{1}{\Gamma(1+g)}\right)^{n}
      G_0^{(N-k)}(1+g)\;\prod_{p=1}^{m-1}\,G_{N-k}^{(N)}(p+(pN-k+1)g),
\end{equation}
where $n=mN-k$. Finally we conclude that the integral \eqref{Selberg-integral} is given by \eqref{Int-g-rep} with $C_m^{(N,\,k)}(g)$ given by \eqref{C-k-answer}. 

Now we are ready to compute the integral $\mathcal{G}_{m}^{(N|\,k_1,k_2,k_3)}(A,B|g)$ defined by \eqref{2D-Selberg-int}. One can use the following observation. For any integral over the plane (more precisely over $n$ two-dimensional planes)
\begin{equation}\label{I}
     \mathbf{I}_{n}=\int\dots\int 
     \prod_{i=1}^{n}|t_{i}|^{2A}|t_{i}-1|^{2B}\prod_{i<j}|t_{i}-t_{j}|^{4g}|P(t_{1},\dots,t_{n})|^{2}\,d^{2}t_{1}\dots d^{2}t_{n},
\end{equation}
where $P(t_{1},\dots,t_{n})$ is arbitrary polynomial the following relation is hold \cite{Dotsenko:1984ad,Dotsenko:1985hi,MR916224}
\begin{equation}\label{I-II}
    \mathbf{I}_{n}=\lambda_{n}(B|g)\,I_{n}\tilde{I}_{n},
\end{equation}
where
\addtocounter{equation}{-1}
\begin{subequations}\label{II}
\begin{align}\label{I-n}
     &I_{n}=\int_{0}^{1}\dots\int_{0}^{1} 
     \prod_{i=1}^{n}t_{i}^{A}(1-t_{i})^{B}\prod_{i<j}|t_{i}-t_{j}|^{2g}\,P(t_{1},\dots,t_{n})\,dt_{1}\dots dt_{n},\\
\intertext{and}\label{I-n-dual} 
     &\tilde{I}_{n}=\int_{1}^{\infty}\dots\int_{1}^{\infty} 
     \prod_{i=1}^{n}t_{i}^{A}(t_{i}-1)^{B}\prod_{i<j}|t_{i}-t_{j}|^{2g}\,P(t_{1},\dots,t_{n})\,dt_{1}\dots dt_{n}.
\end{align}
\end{subequations}
In fact \eqref{I-II} was proven for pure Selberg case (i.e. with $P(t_{1},\dots,t_{n})=1$) by standard manipulations with contour integrals \cite{Lipatov:1976ar}. Since the polynomial is a single-valued function it is not an obstruction to deform contours in a same way as was done for $P(t_{1},\dots,t_{n})=1$. Moreover, since the factor $\lambda_{n}(B|g)$ originates from the non-analyticity of the integrand in \eqref{I} it is going to be the same for any polynomial $P(t_{1},\dots,t_{n})$.    Explicitly $\lambda_{n}(B|g)$ reads as
\begin{equation}\label{Lambda-n}
    \lambda_{n}(B|g)=\frac{(-1)^{n}}{n!}\,\prod_{j=0}^{n-1}\frac{\sin\pi(B+jg)\,\sin\pi(j+1)g}{\sin\pi g}.
\end{equation}
In order to  compute the integral \eqref{I} one has to compute the contour integrals \eqref{I-n} and \eqref{I-n-dual}. The late can be reduced by change of variables $t_{j}\rightarrow 1/t_{j}$ to the integral of the form \eqref{I-n} with some another polynomial $\hat{P}(t_{1},\dots,t_{n})$ defined as
\begin{equation*}
      \hat{P}(t_{1},\dots,t_{n})=\prod_{j=1}^{n}t_{j}^{\eta}\,P(1/t_{1},\dots,1/t_{n}),
\end{equation*}
where $\eta$ is some integer number. In the case considered in this paper polynomials $P(t_{1},\dots,t_{n})$ and $\hat{P}(t_{1},\dots,t_{n})$ belong to the same class. Namely, as follows from \eqref{prop-2}
\begin{equation}\label{P-P-relation}
   \prod_{j=1}^{n}t_{j}^{2(m-1)}\,\mathbf{P}_{m}^{(N|\,k_{1},k_{2},k_{3})}(1/t_{1},\dots,1/t_{n})=
   \mathbf{P}_{m}^{(N|\,k_{3},k_{2},k_{1})}(t_{1},\dots,t_{n}).
\end{equation}
Using \eqref{II}, \eqref{Lambda-n} and \eqref{P-P-relation} one can compute the integral \eqref{2D-Selberg-int}
\begin{multline}\label{2D-Selberg-1}
   \mathcal{G}_{m}^{(N|\,k_1,k_2,k_3)}(A,B|g)=\left(\pi\mu\gamma\left(\frac{bQ}{N}\right)\right)^{n}
   \rho(k_{1},k_{2},k_{3})
   \times\\\times
   \prod_{p=0}^{m-2}\frac{1}
   {\gamma_{N-k}^{(N)}\left(pb^{2}+\frac{(N-k+1)bQ}{N}\right)\,\gamma_{k_{1}}^{(N)}\left(\frac{2b\alpha_{1}+k_{1}}{N}+pb^{2}\right)\,
   \gamma_{k_{2}}^{(N)}\left(\frac{2b\alpha_{2}+k_{2}}{N}+pb^{2}\right)\gamma_{k_{3}}^{(N)}\left(\frac{2b\alpha_{3}+k_{3}}{N}+pb^{2}\right)}
   \times\\
   \frac{1}{\gamma_{0}^{(N-k)}\left(\frac{bQ}{N}\right)\,\gamma_{0}^{(N-k)}\left(\frac{2b\alpha_{1}+k_{1}}{N}+(m-1)b^{2}\right)
   \,\gamma_{0}^{(N-k)}\left(\frac{2b\alpha_{2}+k_{2}}{N}+(m-1)b^{2}\right)\,\gamma_{0}^{(N-k)}\left(\frac{2b\alpha_{3}+k_{3}}{N}+(m-1)b^{2}\right)}
\end{multline}
where integers $k_{1}$, $k_{2}$, $k_{3}$ and $k$ satisfy the condition \eqref{k-def-1}. In \eqref{2D-Selberg-1}  the factor $\rho(k_{1},k_{2},k_{3})$ is defined by \eqref{rho}. We note that the parameters $A$, $B$ and $g$ are expressed in terms of the parameters $\alpha_{j}$ and $b$ as in \eqref{ABg-alpha-relation} while the parameters $\alpha_{1}$, $\alpha_{2}$ and $\alpha_{3}$ are supposed to be related by the screening condition (we recall that $n=mN-k$)
\begin{equation}
    \alpha_{1}+\alpha_{2}+\alpha_{3}+nb=Q.
\end{equation}
In \eqref{2D-Selberg-1} we used also the following notation
\begin{equation}
   \gamma_{k}^{(N)}(x)\overset{\text{def}}{=}
   \prod_{j=1}^{N-k}\gamma\left(x+\frac{(j-1)bQ}{N}\right)\prod_{j=N-k+1}^{N}\gamma\left(x-1+\frac{(j-1)bQ}{N}\right),
\end{equation}
and $\gamma_{0}^{(N-k)}(x)=\prod_{j=1}^{N-k}\gamma(x+(j-1)bQ/N)$.
Now it is easy to show that the function $\mathbb{C}^{k_1,k_2,k_3}(\alpha_1,\alpha_2,\alpha_3)$ defined by eq \eqref{C} satisfies
\begin{equation}\label{C-poles}
  \text{Res}\Bigl|_{\alpha-Q=-nb}\mathbb{C}^{k_1,k_2,k_3}(\alpha_1,\alpha_2,\alpha_3)=
   \mathcal{G}_{m}^{(N|\,k_1,k_2,k_3)}(A,B|g),
\end{equation}
where $\mathcal{G}_{m}^{(N|\,k_1,k_2,k_3)}(A,B|g)$ is given by \eqref{2D-Selberg-1}. 

In the dual case when parameters $k_{1}$, $k_{2}$ and $k_{3}$ are related as in \eqref{k-def-2} corresponding integral can be calculated in a  same way (we present some relevant formulae in the end of appendix \ref{PF}). The resulting integral satisfies the residue condition similar to \eqref{C-poles} with function $\tilde{\mathbb{C}}^{k_1,k_2,k_3}(\alpha_1,\alpha_2,\alpha_3)$ given by \eqref{C-dual}. We note that if $N$ is an odd number then only one condition \eqref{k-def-1} or \eqref{k-def-2} can be satisfied for given $k_{1}$, $k_{2}$ and $k_{3}$ and hence the three-point correlation function \eqref{C-def} (if it is non-zero) is given ether by \eqref{C} or by \eqref{C-dual} as in \eqref{C-prop-odd}. Contrary, for even $N$ both conditions  \eqref{k-def-1} and \eqref{k-def-2} can be valid at the same time. In this case it is natural to expect that the three-point function will be given by sum of \eqref{C} and \eqref{C-dual}.  The phenomenon of this type appears already for $N=2$ i.e. for supersymmetric LFT for the three-point function of two Ramond and one Neveu-Schwarz fields \cite{Poghosian:1996dw}. Our proposal for even $N$ is given by \eqref{C-prop-even}

We note that the residue condition \eqref{C-poles} and similar condition for the function \eqref{C-dual} are not enough to determine the three-point function $C^{k_1,k_2,k_3}(\alpha_1,\alpha_2,\alpha_3)$ completely. Moreover, as one can easily see the proposed exact expression \eqref{C-prop-odd}--\eqref{C-prop-even} has much more poles than predicted by \eqref{C-fun-int-pole}. Our conjecture is inspired by intuition adapted from the bosonic \cite{Zamolodchikov:1995aa} and supersymmetric \cite{Rashkov:1996jx} Liouville field theories and has to be justified by independent methods. In the next section we perform non-trivial check of our conjecture based on dual representation of the parafermionic LFT by the three-exponential model.
\section{Dual representation of parafermionic LFT}\label{3-exp-model}
In this section we consider dual representation of parafermionic LFT based on three-field model suggested in \cite{Fateev:1996ea}. Some notations which we are going to use in this section will overlap with notations used before.

We consider the theory with the action, which contains three exponents
\begin{equation}
  \begin{gathered}
     S_{+-}^{-}=M\exp\{\alpha_{1}\varphi_{1}-\alpha_{2}\varphi_{2}-i\alpha\varphi_{3}\},\quad 
     S_{-+}^{-}=M\exp\{-\alpha_{1}\varphi_{1}+\alpha_{2}\varphi_{2}-i\alpha\varphi_{3}\},\\
     S_{++}^{+}=M\exp\{\alpha_{1}\varphi_{1}+\alpha_{2}\varphi_{2}+i\alpha\varphi_{3}\} \label{V}
  \end{gathered}  
\end{equation}
where parameters $\alpha_{1},\alpha_{2}$ and $\alpha$ satisfy the condition:
\begin{equation}\label{c}
\alpha^{2}-\alpha_{1}^{2}-\alpha_{2}^{2}=1/2. 
\end{equation}
Namely
\begin{equation}
A=\int d^{2}x\left(  \frac{1}{16\pi}\sum_{i}^{3}(\partial_{\mu}\varphi
_{i})^{2}+S_{+-}^{-}+S_{-+}^{-}+S_{++}^{+}\right)  \label{a}
\end{equation}
The affine version of the theory \eqref{a} where the term $S^{+}_{--}$ is added to the action was considered in different regimes\footnote{There are three regimes in the affine theory when all parameters $\alpha$ and $\alpha_{j}$ are imaginary, when $\alpha$ is real and both $\alpha_{j}$ are imaginary and finally the regime considered in this paper when all parameters are real. In the first regime when all parameters $\alpha$ and $\alpha_{j}$ are imaginary the theory is unitary and its spectrum consists of quadruplet of fundamental particles and their bound states \cite{Fateev:1996ea}. The theory is integrable with two-particle $S$-matrix being proportional to the tensor product of two-soliton $S$-matrices of the sine-Gordon model. The integral representations for the form-factors of local fields in this theory were studied in \cite{Fateev:2004un}. The particle content and scattering theory in the remaining non-unitary regimes of affine theory are rather complicated and deserves further studies.} in  \cite{Fateev:1996ea,Baseilhac:1998eq,Fateev:2004un}. The theory \eqref{a} is non-unitarian CFT with the stress energy tensor \cite{Fateev:1996ea,Feigin:2001yq}
\begin{equation}
   T(z)=-\frac{1}{4}\sum_{i}^{3}(\partial_{z}\varphi_{i})^{2}+\frac{\partial_{z}^{2}\varphi_{1}}{4\alpha_{1}}+
   \frac{\partial_{z}^{2}\varphi_{2}}{4\alpha_{2}}+\frac{i\partial_{z}^{2}\varphi_{3}}{4\alpha} \label{t}
\end{equation}
and central charge $c=3-\frac{6}{4\alpha^{2}}+\frac{6}{4\alpha_{1}^{2}}+\frac{6}{4\alpha_{2}^{2}}$. It is convenient to write our parameters
$\alpha$, $\alpha_{j}$ in the form
\begin{equation}
     4\alpha^{2}=N+2,\quad4\alpha_{1}^{2}=\frac{Nb^{2}}{1+b^{2}},\quad4\alpha_{2}^{2}=\frac{N}{1+b^{2}}. \label{p}
\end{equation}
Then the central charge will have a form
\begin{equation}\label{cc}
   c=\frac{3N}{N+2}+\frac{6Q^{2}}{N}, 
\end{equation}
which is exactly the same as \eqref{central-charge}. The primary fields will be exponential operators
\begin{equation}
  V(\vec{a})=\exp\left(  a_{1}\varphi_{1}+a_{2}\varphi_{2}+ia_{3}\varphi_{3}\right)  \label{pf}
\end{equation}
with conformal dimensions
\begin{equation*}
  \Delta(\vec{a})=a_{1}\left(\frac{1}{2\alpha_{1}}-a_{1}\right)+a_{2}\left(\frac{1}{2\alpha_{2}}-a_{2}\right)-a_{3}\left(\frac{1}{2\alpha}-a_{3}\right).
\end{equation*}
The correlation functions of exponential fields can be calculated using Goulian-Lie method described above
\begin{multline}\label{si}
  \left\langle V(\vec{a}_{1})(x_{1})...V(\vec{a}_{r})(x_{r})\right\rangle=
  \frac{\Gamma(n_{2}-n_{1}-n_{3})}{\alpha_{1}}\frac{\Gamma(n_{1}-n_{2}-n_{3})}{\alpha_{2}}\frac{\Gamma(n_{1}+n_{2}-n_{3})}{\alpha}
  \times\\\times
  \left\langle V(\vec{a}^{(1)})(x_{1})...V(\vec{a}^{(r)})(x_{r})(-\mathbf{S}_{+-}^{-})^{n_{1}}(-\mathbf{S}_{-+}^{-})^{n_{2}}(-\mathbf{S}_{++}^{+})^{n_{3}}\right\rangle _{FF}
\end{multline}
where $\mathbf{S}_{+-}^{-}$ $=M\int d^{2}xS_{+-}^{-}(x)...$ and index $FF$ denotes that expectation values are taken over free field vacuum. The numbers $n_{i}$ are restricted by the conditions
\begin{equation*}
  \begin{gathered}
   \sum_{j=1}^{r}a_{1}^{(j)}+(n_{1}-n_{2}+n_{3})\alpha_{1} =\frac{1}{2\alpha_{1}},\quad
   \sum_{j=1}^{r}a_{2}^{(j)}+(n_{2}-n_{1}+n_{3})\alpha_{1} =\frac{1}{2\alpha_{2}},\\
   \sum_{j=1}^{r}a_{3}^{(j)}+(n_{3}-n_{1}-n_{2})\alpha=\frac{1}{2\alpha}.
  \end{gathered} 
\end{equation*}
Free-field correlation function in the r.h.s. in \eqref{si} can be expressed in terms of Coulomb integral only if all the parameters $n_{i}$ are non-negative integers otherwise eq \eqref{si} is rather formal. The calculation of the correlation functions in our theory is very similar to the same calculation for sine-Liouville model \cite{Fateev-unpublished} (see also \cite{Ribault:2005wp}) and based on usage of certain integral identity (see for example eq (1.13) in \cite{Fateev:2007qn}). We prefer to hide details here and present only final answer\footnote{Interested reader will be able to obtain \eqref{corf} with the help of integral identity (1.13) from \cite{Fateev:2007qn} in a simple way.}. Let us suppose the condition
\begin{equation}
   \sum_{j=1}^{r}a_{3}^{(j)}-\frac{1}{2\alpha}=r^{\prime}\alpha\label{r'}
\end{equation}
with some integer number $r'$ and define $p$ as 
\begin{equation}
   p=r-r^{\prime}-2,
\end{equation}
then correlation function of the fields $V(\vec{a}_{j})(x_{j})$ can be expressed in terms of correlation functions of two Liouville field theories with coupling constants $2\alpha_{1}$ and $2\alpha_{2}$ with cosmological constants $-\mu^{2}\gamma(-4\alpha_{1}^{2})$ and $-\mu^{2}\gamma(-4\alpha_{2}^{2})$ (we label corresponding correlation functions with lower subscripts $1$ and $2$ respectively). Namely,
\begin{multline}\label{corf}
  \langle V(\vec{a}_{1})(x_{1})...V(\vec{a}_{r})(x_{r})\rangle=\frac{(\pi M)^{r^{\prime}}}{\Gamma(1+p)}
  \prod\limits_{i=1}^{r}\mathcal{N}(\vec{a}_{j})
  \int d^{2}t_{1}\dots d^{2}t_{p}\;\Omega(t_{1},\dots,t_{p}|x_{1},\dots,x_{r})\\
  \langle V_{-\alpha_{1}}(t_{1})\dots V_{-\alpha_{1}}(t_{p})V_{a_{1}^{(1)}+\alpha_{1}}(x_{1})\dots V_{a_{1}^{(r)}+\alpha_{1}}(x_{r})\rangle_{1}
  \langle V_{-\alpha_{2}}(t_{1})\dots V_{-\alpha_{2}}(t_{p})V_{a_{2}^{(1)}+\alpha_{2}}(x_{1})\dots V_{a_{2}^{(r)}+\alpha_{2}}(x_{r})\rangle_{2}
\end{multline}
where $V_{a_{i}}$ are the standard exponential operators in Liouville field theory with conformal dimensions $a_{i}(Q-a_{i})$ and factors $\mathcal{N}(\vec{a}_{j})$ are
\begin{equation}
   \mathcal{N}(\vec{a}_{j})=\gamma(1+2\alpha a_{3}^{(j)}-2\alpha_{1}a_{1}^{(j)}-2\alpha_{2}a_{2}^{(j)}). \label{NA}
\end{equation}
The factor $\Omega(t_{1},\dots,t_{p}|x_{1},\dots,x_{r})$ in equation (\ref{corf}) is given by
\begin{equation}\label{fac}
  \Omega(t_{1},\dots,t_{p}|x_{1},\dots,x_{r})=
 \frac{{\displaystyle\prod\limits_{i<j}^{p}}|t_{i}-t_{j}|^{4\alpha^{2}}{\displaystyle\prod\limits_{k<m}^{r}}
        |x_{k}-x_{m}|^{4(a_{3}^{(k)}+\alpha)(a_{3}^{(m)}+\alpha)}}
        {{\displaystyle\prod\limits_{i}^{p}}{\displaystyle\prod\limits_{m}^{r}} |t_{i}-x_{m}|^{4\alpha(a_{3}^{(m)}+\alpha)}}.
\end{equation}
We see that our model, similar to sine-Liouville model admits the separation of variables and  possesses the dual sigma model representation (see appendix \ref{MSP-approximation}). The general three-point correlation function in our CFT is rather complicated object which we suppose to study in separate publication. However it is rather easy to calculate the two point functions which coincide with special reflection coefficients of the fields $V(\vec{a}).$ We normalize our fields $V(\vec{a})$
by the condition
\begin{equation}\label{NC}
   \left\langle V(a_{1},a_{2},a_{3})V\left(  \frac{1}{2\alpha_{1}}-a_{1},\frac{1}{2\alpha_{2}}-a_{2},\frac{i}{2\alpha}-a_{3}\right)  \right\rangle=1
\end{equation}
Then the reflection amplitude defined as $V(a_{1},a_{2},a_{3})=R_{1}(\vec
{a})V\left(  \frac{1}{2\alpha_{1}}-a_{1},a_{2},a_{3}\right)  $ coincides with
the two point function
\begin{equation}
R_{1}(a_{1},a_{2},a_{3})=\left\langle V(a_{1},a_{2},a_{3})V\left(  a_{1}%
,\frac{1}{2\alpha_{2}}-a_{2},\frac{i}{2\alpha}-a_{3}\right)  \right\rangle
.\label{R1}
\end{equation}
In the similar way can be defined amplitude $R_{2}(\vec{a})$ and amplitude
associated with double reflection
\[
V(a_{1},a_{2},a_{3})=R_{1,2}(a_{1},a_{2},a_{3})V\left(  \frac{1}{2\alpha_{1}%
}-a_{1},\frac{1}{2\alpha_{2}}-a_{2},a_{3}\right)
\]
which coincides with the  two-point function
\begin{equation}
R_{1,2}(a_{1},a_{2},a_{3})=\left\langle V(a_{1},a_{2},a_{3})V\left(
a_{1},a_{2},\frac{i}{2\alpha}-a_{3}\right)  \right\rangle \label{R12}%
\end{equation}
The amplitude associated with full reflection
\[
V(a_{1},a_{2},a_{3})=R(\vec{a})V\left(  \frac{1}{2\alpha_{1}}-a_{1},\frac
{1}{2\alpha_{2}}-a_{2},\frac{i}{2\alpha}-a_{3}\right)
\]
is equal to the two-point function
\begin{equation}
R(a_{1},a_{2},a_{3})=\left\langle V(a_{1},a_{2},a_{3})V\left(  a_{1}%
,a_{2},a_{3}\right)  \right\rangle \label{R123}%
\end{equation}
It is easy to derive from eqs.(\ref{corf},\ref{NA}) that functions
$R_{1},R_{2}$ and $R_{1,2}$ have the form:
\begin{align}
  &R_{1}(\vec{a})=\frac{R_{2\alpha_{1}}^{L}(a_{1}+\alpha_{1})}
  {\gamma(2\alpha_{1}a_{1}+2\alpha_{2}a_{2}-2\alpha a_{3})\gamma(2\alpha_{1}a_{1}-2\alpha_{2}a_{2}+2\alpha a_{3})} \label{RR},\\
  &R_{2}(\vec{a})=\frac{R_{2\alpha_{2}}^{L}(a_{2}+\alpha_{2})}{\gamma(2\alpha_{1}a_{1}+2\alpha_{2}a_{2}-2\alpha a_{3})
  \gamma(-2\alpha_{1}a_{1}+2\alpha_{2}a_{2}+2\alpha a_{3})} \label{RRR}
\end{align}
and
\begin{equation}
R_{1,2}(\vec{a})=\frac{R_{2\alpha_{1}}^{L}(a_{1}+\alpha_{1})R_{2\alpha_{1}%
}^{L}(a_{2}+\alpha_{2})}{\gamma(2\alpha_{1}a_{1}+2\alpha_{2}a_{2}-2\alpha
a_{3})\gamma(-1+2\alpha_{1}a_{1}+2\alpha_{2}a_{2}+2\alpha a_{3})} \label{Rr}.
\end{equation}
Here the functions $R_{2\alpha_{i}}^{L}$ are the reflection amplitudes (two-point correlation functions in the Liouville CFT with coupling constants $2\alpha_{i}$ and cosmological constants $\mu_{i}=-\pi M^{2}\gamma(-4\alpha_{i}^{2})$)
\begin{equation}
  R_{2\alpha_{i}}^{L}(a_{i}+\alpha_{i})=\left(  \frac{\pi M}{4\alpha_{i}^{2}}\right)  ^{(\frac{1}{\alpha_{i}}-4a_{i})/2\alpha_{i}}
  \frac{\Gamma(4\alpha_{i}(a_{i}-\frac{1}{4\alpha_{i}})\Gamma(1+\frac{1}{\alpha_{i}}(a_{i}-\frac{1}{4\alpha_{i}})}
  {\Gamma(-4\alpha_{i}(a_{i}-\frac{1}{4\alpha_{i}})\Gamma(1-\frac{1}{\alpha_{i}}(a_{i}-\frac{1}{4\alpha_{i}})}.\label{RL}%
\end{equation}
It is easy to see that reflection amplitudes satisfy the  ``associativity condition''
\[
  R_{1,2}(a_{1},a_{1},a_{3})=R_{1}(a_{1},a_{2},a_{3})R_{2}(\frac{1}{2\alpha_{1}}-a_{1},a_{2},a_{3})
\]
The reflection amplitude $R_{3}$ defined as $V(a_{1},a_{2},a_{3})=R_{3}(\vec{a})V\left(  a_{1},a_{2},\frac{i}{2\alpha}-a_{3}\right)  $ coincides with
correlation function
\[
R_{3}(\vec{a})=\left\langle V(a_{1},a_{2},a_{3})V\left(  \frac{1}{2\alpha_{1}%
}-a_{1},\frac{1}{2\alpha_{2}}-a_{2},a_{3}\right)  \right\rangle .
\]
It can be calculated with the result
\begin{multline*}
  R_{3}\left(  \vec{a}\right)=\left(  \frac{\pi M}{4\alpha^{2}}\right)^{(\frac{1}{\alpha}-4a_{3})/2\alpha}\frac
  {\Gamma(-4\alpha(a_{3}-\frac{1}{4\alpha}))
  \Gamma(1+\frac{1}{\alpha}(a_{3}-\frac{1}{4\alpha}))}{\Gamma(4\alpha(a_{3}-
  \frac{1}{4\alpha}))\Gamma(1-\frac{1}{\alpha}(a_{3}-\frac{1}{4\alpha}))}\times\\\times
 \gamma(1-\alpha_{3}a_{3}+\alpha_{2}a_{2}-\alpha a_{3})\gamma(1-\alpha_{3}a_{3}-\alpha_{2}a_{2}+\alpha a_{3}).
\end{multline*}
The amplitude $R\left(  \vec{a}\right)  =R_{3}\left(  \vec{a}\right)R_{1,2}\left(  \vec{a}\right)$. If we introduce the momentums $p_{1}$, $p_{2}$  and $q$ by the relation
\begin{equation*}
   2\alpha_{j}a_{j}=\frac{1}{2}+i2\alpha_{j}p_{j},\qquad 2\alpha a_{3}=\frac{1}{2}+2\alpha q,
\end{equation*} 
then denominator in eq.(\ref{Rr}) will have a form
\[
\gamma(\frac{1}{2}+2i\alpha_{1}p_{1}+2i\alpha_{2}p_{2}-2\alpha q)\gamma
(\frac{1}{2}+2i\alpha_{1}p_{1}+2i\alpha_{2}p_{2}+2\alpha q),
\]
and it easy to see that amplitude $R_{1,2}$ will be unitarian with respect to $p_{j}\rightarrow-p_{j}$. 

At this point we will not discuss here any more the general CFT (\ref{a}) and start to consider the special subspace of our model which is related with parafermionic Liouville model considered above. We have seen already that the central charge of our model \eqref{cc} coincides with the central charge of parafermionic Liouville model \eqref{central-charge}. We show that certain three-point correlation functions \eqref{C-def} in parafermionic LFT can be expressed in terms of special correlation functions of our model.  It is well known that  the parafermionic algebra admits the free field representation in terms of two fields $\phi_{1}$ and $\phi_{2}$ with the stress energy tensor \cite{Griffin:1988tf}
\begin{equation}
 T_{\text{PF}}(z)=-\frac{1}{4}(\partial_{z}\phi_{1})^{2}-\frac{1}{4}(\partial_{z}\phi_{2})^{2}+\frac{i}{4\alpha}\partial_{z}^{2}\phi_{2}. \label{TP}%
\end{equation}
Parafermionic currents $\Psi(z)$ and $\Psi^{+}(z)$ can be represented in terms of the chiral parts of these fields as \cite{Griffin:1988tf}
\begin{equation}\label{FFP}
  \begin{aligned}
  &\Psi(z)=\left(  \frac{N+2}{4N}\right)  ^{1/2}\left(  i\partial_{z}\phi_{2}+\left(  \frac{N}{N+2}\right)  ^{1/2}\partial_{z}\phi_{1}\right)
     \exp\left(  -\phi_{1}(z)/\sqrt{N}\right),\\
   &\Psi^{+}(z)=\left(  \frac{N+2}{4N}\right)  ^{1/2}\left(  i\partial_{z}\phi_{2}-\left(  \frac{N}{N+2}\right)  ^{1/2}\partial_{z}\phi_{1}\right)
\exp\left(  \phi_{1}(z)/\sqrt{N}\right)
\end{aligned}
\end{equation}
The order fields $\sigma_{k}$ can can be represented in terms of these fields as
\begin{equation}
\sigma_{k}=\mathbf{N}_{k}\exp\left(  -\frac{k\phi_{1}}{2\sqrt{N}}\right)  \exp\left(
-\frac{ik\phi_{2}}{4\alpha}\right),  \label{sk}
\end{equation}
where $\mathbf{N}_{k}$ is the normalization factor which will be specified below.
To establish relation between two models it is important to express fields
$\phi_{i}$, $\varphi$ from \eqref{a} in terms of the fields $\varphi_{i}$ from \eqref{TP} and vice versa. For that it is convenient to compare the stress energy tensor (\ref{t}) with stress
energy tensor of parafermionic Liouville model
\begin{equation}\label{tpfl}
  T_{\text{PFL}}=T_{\text{PF}}-N(\partial\varphi)^{2}+Q \partial^{2}\varphi
\end{equation}
It is easy to see that field $\phi_{2}$ can be identified with field $\varphi_{3}$. For other fields we derive
\begin{equation}\label{rel}
\begin{aligned}
   &\varphi_{1}=4\alpha_{2}\varphi-\frac{2\alpha_{1}\phi_{1}}%
  {\sqrt{N}},\quad&&\varphi_{2}=4\alpha_{1}\varphi+\frac{2\alpha_{2}\phi_{1}}{\sqrt{N}},\\
  &\varphi=\frac{\alpha_{2}\varphi_{1}}{N}+\frac{\alpha_{1}\varphi_{2}}{N},\quad
  &&\phi_{1}=\frac{2\alpha_{1}\varphi_{1}}{\sqrt{N}}-\frac{2\alpha_{2}\varphi_{2}}{\sqrt{N}}.
\end{aligned}
\end{equation}
It is interesting to express the screening charges in parafermionic Liouville model which are 
\begin{equation*}
  J_{1}=\Psi(z)e^{2b\varphi(z)},\qquad J_{2}=\Psi^{+}(z)e^{2b\varphi(z)}, 
\end{equation*}
in terms of fields $\varphi_{i}.$ We derive using the relation $b=\frac{\alpha_{1}}{\alpha_{2}}$ that
\begin{equation}\label{J}
\begin{aligned}
  &J_{1}=\left(  \frac{1}{4N}\right)  ^{1/2}(i\alpha\partial_{z}\varphi_{3}+\alpha_{1}\partial_{z}\varphi_{1}-\alpha_{2}\partial_{z}\varphi_{2}%
     )\exp\left(  \frac{\varphi_{2}(z)}{2\alpha_{2}}\right),\\
  &J_{2}=\left(  \frac{1}{4N}\right)  ^{1/2}(i\alpha\partial_{z}\varphi_{3}-\alpha_{1}\partial_{z}\varphi_{1}+\alpha_{2}\partial_{z}\varphi_{2}%
     )\exp\left(\frac{\varphi_{1}(z)}{2\alpha_{1}}\right).
\end{aligned}
\end{equation}
It is easy to check that fields $J_{1}$ and $J_{2}$ commute with all screening
charges $S_{+-}^{-}...$ given by eq(\ref{V}) and operators $\int d^{2}x\mu
_{i}^{\prime}J_{i}\overline{J_{i}}$ can be used as additional screening
charges for calculation of the correlation functions. We note that there is
one additional screening field $J_{3}$%
\begin{equation}
J_{3}=\left(  \frac{1}{4N}\right)  ^{1/2}(-i\alpha\partial_{z}\varphi
_{3}+\alpha_{1}\partial_{z}\varphi_{1}-\alpha_{2}\partial_{z}\varphi_{2}%
)\exp\left(  \frac{i\varphi_{3}(z)}{2\alpha}\right).  \label{j3}%
\end{equation}

Now we can consider the simplest three point correlation function in parafermionic Liouville model. Namely, $C^{k_1,k_2,k_3}(a_1,a_2,a_3)$ where $k_{1}+k_{2}+k_{3}=N$ and take for simplicity $N$ odd. In this case $r^{\prime}=-1$ and hence $p=0$ so our correlation
function does not contain any integrals. It follows from eqs \eqref{rel} and \eqref{sk} that the field $V_{a}^{(k)}$ has a form
\begin{equation*}
   V_{a}^{(k)}=\mathbf{N}_{k}\exp\left(  -\frac{k(\alpha_{1}\varphi_{1}-\alpha_{2}\varphi_{2})}{N}\right)  \exp\left(  -\frac{ik\varphi_{3}}{4\alpha}\right)
   \exp\left(  a\frac{(2\alpha_{2}\varphi_{1}+2\alpha_{1}\varphi_{2})}{N}\right)  .
\end{equation*}
The normalization factors $\mathcal{N}(\vec{a})$ in eq (\ref{corf}) in this case are
\[
\mathcal{N}(\vec{a}_{i})=\mathcal{N}_{k}(a_{i})=\frac{1}{\gamma((8\alpha_{1}\alpha_{2}a_{i}+4k_{i}\alpha_{2}^{2})/N)}%
\]
and the there-point correlation function $C^{k_1,k_2,k_3}(a_1,a_2,a_3)$ with $k_{1}+k_{2}+k_{3}=N$ can be written in the form
\begin{equation}
C^{k_1,k_2,k_3}(a_1,a_2,a_3)=
\prod_{i=1}^{3}\mathbf{N}_{k_{i}}\,\mathcal{N}_{k_{i}}(a_{i})\times C_{2\alpha_{1}}(A_{1},A_{2},A_{3}%
)C_{2\alpha_{2}}(B_{1},B_{2},B_{3}) \label{tpf}%
\end{equation}
where $C_{2\alpha_{i}}$ denotes the three-point correlation function in Liouville CFT  with the
coupling constant $2\alpha_{i}$ and
\[
A_{i}=\frac{2\alpha_{2}a_{i}+\alpha_{1}(N-k_{i})}{N};\quad B_{i}=\frac
{2\alpha_{1}a_{i}+\alpha_{2}(N+k_{i})}{N}.
\]
One can use the following  easily established relations
\begin{multline}
\mathcal{N}_{k}(a)  \Upsilon_{2\alpha_{1}}\left(  \frac{4\alpha_{2}a+2\alpha_{1}k}%
{N}\right)  \Upsilon_{2\alpha_{2}}\left(  \frac{4\alpha_{1}a+2\alpha_{2}%
(N+k)}{N}\right)=\\=
\kappa\Upsilon_{2\alpha_{1}}\left(  \frac{4\alpha_{2}a+2\alpha_{1}(N-k)}%
{N}\right)  \Upsilon_{2\alpha_{2}}\left(  \frac{4\alpha_{1}a+2\alpha_{2}k}%
{N}\right)  \label{ka}%
\end{multline}
and
\begin{multline}\label{kaa}
   \Upsilon_{N-k}^{(N)}(a,b)=
   \left(2\alpha_{2}\right)  ^{-\left(  \frac{2\alpha_{2}a+2\alpha_{1}k}{N}-\frac{\alpha_{1}^{2}+1}{\alpha_{1}}\right)  }\left(  2\alpha_{1}\right)  ^{-\left(
   \frac{2\alpha_{1}a+2\alpha_{2}(N-k)}{N}-\frac{\alpha_{2}^{2}+1}{\alpha_{2}}\right)}\times\\\times
   \Upsilon\left(\frac{2\alpha_{2}a+2\alpha_{1}k}{N},2\alpha_{1}\right)  \Upsilon\left(\frac{2\alpha_{1}a+2\alpha_{2}(N-k)}{N},2\alpha_{2}\right),
\end{multline}
where $\kappa=\left(  2\alpha_{2}\right)  ^{-\frac{16\alpha_{2}\alpha
_{1}a+8\alpha_{2}^{2}k}{N}+1}$,  $\alpha_{1}$ and $\alpha_{2}$ are taken as in  \eqref{p}, $\Upsilon_{k}^{(N)}(x,b)$ is given by  \eqref{Upsilon-N-def} and $\Upsilon(x,b)$ is ordinary Upsilon-function. Using \eqref{ka} and \eqref{kaa} we find that if we take into account that cosmological constants are equal to $-\gamma(-4\alpha_{i}^{2})M^{2}$ as well as the relation between the parameters $\mu$ in parafermionic LFT and $M$ in three-exponential model derived in \cite{Baseilhac:1998eq}
\begin{equation}
 \pi\mu\gamma\left(  \frac{bQ}{N}\right)  =\left(  \frac{\pi M}{4\alpha_{2}^{2}}\right)  ^{2\frac{bQ}{N}} \label{ba}%
\end{equation}
where $\frac{bQ}{N}=\frac{1}{4\alpha_{2}^{2}}$ and also the explicit form of normalization factors $\mathbf{N}_{k}$
\begin{equation*}
   \mathbf{N}_{k}^{2}=\gamma\left(\frac{N-k+1}{N+2}\right),
\end{equation*}
computed in \cite{Baseilhac:1998eq} we obtain that \eqref{tpf} coincides exactly with the expression \eqref{C-dual} for the three-point correlation function proposed above in the case $k_{1}+k_{2}+k_{3}=N$. 

For the calculation in the more general case when parameters $k_{j}$ are related by the condition $k_{1}+k_{2}+k_{3}=N+2k$ we should introduce $k$ additional screening charges generated by the field $J_{3}$ (\ref{j3}). The calculations become more involved but they give exactly the same answer as in \eqref{C-dual}.

Now we can calculate the reflection amplitude for parafermionic Liouville field theory. The simple calculation gives a result that function defined by
the relation
\begin{equation}
V_{a}^{(k)}=R_{k}(a)V_{Q-a}^{(N-k)} \label{R}%
\end{equation}
is equal
\[
R_{k}(a)=\left(  \frac{\pi\mu}{b^{4/N}}\gamma\left(  \frac{bQ}%
{N}\right)  \right)  ^{\frac{2a-Q}{b}}\gamma\left(\frac{b(2a-Q)+k}{N}\right)\gamma
\left(\frac{b^{-1}(2a-Q)+(N-k)}{N}\right)b^{2\frac{N-2k}{N}}.
\]
These functions satisfy the condition $R_{k}(a)R_{N-k}(Q-a)=1.$ For arbitrary
$b$ functions $R_{0}(a)=R_{N}(a)$ and $R_{\frac{N}{2}}(a)$ (for even $N)$
satisfy unitarity conditions $R_{0}(a)R_{0}(Q-a)=R_{\frac{N}{2}}(a)R_{\frac
{N}{2}}(Q-a)=1.$ For $b=1$ all functions $R_{k}(a)$ satisfy this condition and
full theory is unitary. Using this reflection amplitude we can derive all three-point functions from the three-point function with $k_{1}+k_{2}+k_{3}=N+2k$. In this way we have an alternative derivation of the result obtained in section \ref{integrals}.
\section{Concluding remarks and open problems}\label{conclusions}
In this section we discuss briefly some subjects and questions which were not
considered in the main body of the paper.
\begin{enumerate}
\item The important question in the analysis of any CFT is the description of the
chiral algebra of theory. Besides the non-local symmetry generated by the field
(\ref{G}) our theory possesses the infinite dimensional local symmetry generated
by holomorphic $W-$current $W_{4}$ and other currents with even spin which
appear in the operator product expansion of this current with itself and with
these new holomorphic currents\footnote{The field $W_{4}$  and other
$W-$currents appear in the OPE of two fields (\ref{G}) in the sector
corresponding to identity operator.}. These currents form the closed $W$
algebra. For three exponential model considered in Section 4 the explicit form of
these currents in terms of three fields can be found in \cite{Fateev:1996ea,Baseilhac:1998eq,Feigin:2001yq}. The fields
\begin{equation}
V_{\alpha}^{(k,q)}=\Phi_{(q,q)}^{k}e^{2\alpha\varphi} \label{VQ}%
\end{equation}
where fields $\Phi_{(q,q)}^{k}$ are defined in appendix A are the primary
fields of this $W-$ algebra. The full space of fields in  parafermionic LFT, local with
respect the fields $V_{\alpha}^{k}$ can be derived by application of this $W$
algebra to the  primary fields
\begin{equation}
F=\sum_{k=0}^{N-1}\sum_{q=0}^{N-1}\left[  V_{\alpha}^{(k,q)}\right]
_{W,\overline{W}} \label{F}%
\end{equation}
\item In this paper we calculated the three-point correlation functions for the fields
$V_{\alpha}^{k}.$ Three point correlation functions of more general local fields
$V_{\alpha}^{(k,q)}=\Phi_{(q,q)}^{k}e^{2\alpha\varphi}$
\[
\,\,\left\langle V_{\alpha_{1}}^{(k_{1},q_{1})}(0)V_{\alpha_{2}}^{(k_{2}%
,q_{2})}(1)V_{\alpha_{3}}^{(k_{3},q_{3})}(\infty)\right\rangle
\]
can be calculated by generalization of our method for this case. The analog of
the polynomials $\mathbf{P}^{(N|k_{1},k_{2}k_{3})}$ for this case can be
derived by the proper fusions of the points in this polynomials dictated by
the parafermionic algebra. Namely, any field $\Phi_{(q,q)}^{k}$ can be
obtained in the operator product expansion of \textquotedblleft
order\textquotedblright\ field $\sigma_{k}$ with parafermionic currents $\Psi$
and $\overline{\Psi}.$ As an example we give here the three-point correlation functions of two
operators $V_{\alpha_{i}}^{(k)}$ with operator $V_{\alpha}^{(k,q)}$ $\ $which
determine the operator product expansion of two fields $V_{\alpha_{i}}^{(k_{i})}$
(with $k_{1}\geq k_{2}$ and $k_{1}+k_{2}<N$)
\begin{equation}
V_{\alpha_{1}}^{k_{1}}V_{\alpha_{2}}^{k_{2}}=\sum_{\alpha}\sum_{j=k_{1}-k_{2}%
}^{k_{1+}k_{2}}\sum_{q=0}^{N-1}C_{\alpha,k_{1}-k_{2}+2j,q}^{(\alpha_{1}%
,k_{1}|\alpha_{2}k_{2})}V_{\alpha}^{(k_{1}-k_{2}+2j,q)}. \label{SC}%
\end{equation}
The structure constant $C_{\alpha,k,q}^{(\alpha_{1},k_{1}|\alpha_{2}k_{2})}$
\ in this expansion can be expressed through the three-point correlation function $\langle V_{\alpha_{1}}^{k_{1}}(0)V_{\alpha_{2}}^{k_{2}}(1)V_{\alpha
}^{(k,k+2q)}(\infty)\rangle$. For $k_{1}+k_{2}+k=2l<N$ it has a form
\begin{multline}
\left\langle V_{\alpha_{1}}^{k_{1}}(0)V_{\alpha_{2}}^{k_{2}}(1)V_{\alpha
}^{(k,k+2q)}(\infty)\right\rangle =A(\alpha,k,q,b)
 \left[\pi\mu\gamma\left(\frac{bQ}{N}\right)b^{-\frac{2bQ}{N}}\right]^{\frac{Q-\alpha}{b}}
    \rho(k_{1},k_{2},k)\times\\\times
    \frac{\Upsilon_{0}^{(N)}(b)\Upsilon_{k_1}^{(N)}(2\alpha_1)\Upsilon_{k_2}^{(N)}(2\alpha_2)\Upsilon_{k+2q}^{(N)}(2\alpha_3)}
    {\Upsilon_{l+q}^{(N)}(\alpha_1+\alpha_2+\alpha_3-Q)\Upsilon_{l-k_1+q}^{(N)}(\alpha_2+\alpha_3-\alpha_1)
    \Upsilon_{l-k_2+q}^{(N)}(\alpha_1+\alpha_3-\alpha_2)\Upsilon_{l-k-q}^{(N)}(\alpha_1+\alpha_2-\alpha_3)},
\end{multline}
here function $A(\alpha,k,q,b)$ for $0<q\leq N-k$ has a form
\[
A=\frac{(N-k-q)!N!^{2}}{q!(N-k)!(N-q)!^{2}}\,\,
{\displaystyle\prod\limits_{j=0}^{q-1}}
\frac{(-N^{2})^{j}}{(2\alpha+(N-k-2q)b+jQ)^{2}};\quad0<2q\leq N-k
\]
and for $N-k<2q\leq2N-2k$
\[
A=\frac{(N-k-q)!N!^{2}}{q!(N-k)!(N-q)!^{2}}
{\displaystyle\prod\limits_{j=1}^{N-k-q}}
\frac{(-N^{2})^{j}}{(2\alpha+(N-k-2q)b+(N-k-j-1)Q)^{2}};
\]
For $q>N-k$ \ the function $A(\alpha,k,q,b)$ can be calculated using the
\textquotedblleft duality\textquotedblright\ relation 
$$
A(\alpha,k,q,b)=A(\alpha,N-k,N-q,\frac{1}{b}).
$$
\item The important role in studying of parafermionic LFT play the degenerate fields
$V_{-nb/2-m/2b}^{k}$ where non-negative integers $n,m$ satisfy the relation
$m-n=k$ $(\operatorname{mod}N).$ The OPE of fields $V_{\alpha}^{k}$ (and
$V_{\alpha}^{(k,q)}$) with these fields give only finite number of field
$V_{\alpha^{\prime}}^{(k^{\prime},q^{\prime})}.$ The four point correlation
functions of the degenerate field with three arbitrary fields $V_{\alpha_{i}%
}^{k_{i}}$ can be represented by the finite dimensional integrals and satisfy the differential equations. As an example we give here the simplest four
point functions in parafermionic LFT. Let $i+j+k=l$ \ where $2l<N.$ Then
\[
\left\langle V_{\alpha_{3}}^{i}(\infty)V_{\alpha_{2}}^{j}(1)V_{\alpha_{1}}%
^{k}(0)V_{-(N-l)b/2}^{l}(x)\right\rangle =\text{const}\,|z|^{2\delta_{1}}%
|1-z|^{2\delta_{2}}\mathbf{J}_{N-l}(A_{k},B_{j},C_{l}|x)
\]
where $\delta_{1}=\frac{\left(  N-l\right)  }{N}(2\alpha_{1}b+\frac
{k(1+N)}{\left(  N+2\right)  })$, $\delta_{2}=\frac{\left(  N-l\right)  }%
{N}(2\alpha_{2}b+\frac{j(1+N)}{\left(  N+2\right)  })$ and
\[
\mathbf{J}_{N-l}(A_{k},B_{j},C_{l}|x)=\int
{\displaystyle\prod\limits_{p=1}^{N-l}}
d^{2}t_{p}|t_{p}|^{2A_{k}}|t_{p}-1|^{2B_{j}}|t_{p}-x|^{2C_{l}}
{\displaystyle\prod\limits_{m<n}^{N-l}}
|t_{m}-t_{n}|^{-4bQ/N}
\]
with 
\begin{equation*}
\begin{gathered}
  A_{k}=\frac{b(\alpha-2\alpha_{1}-Q-k/b-(N-l)b/2)}{N};\quad B_{j}=\frac{b(\alpha-2\alpha_{2}-Q-j/b-(N-l)b/2)}{N};\\
  C_{l}=\frac{b(Q-\alpha-l/b+(N-l)b/2)}{N},
\end{gathered}
\end{equation*}
here $\alpha=\alpha_{1}+\alpha_{2}+\alpha_{3}$.
Using the results of paper \cite{Fateev:2009me} where the correspondence between the
integral $\mathbf{J}_{m}(A,B,C|x)$ and differential equation of order $m+1$
was established it is easy to derive the differential equations of order
$N-l+1$ for our correlation functions.
\item The degenerate fields play even more important role if we consider our
theory in the region $b^{2}<0.$ For $|b|<1$ parameter \ $\alpha_{1}^{2}$ in eq
(\ref{p}) becomes negative and we are in the other regime of the three exponential model. It
is convenient to take $\ -b^{2}=\frac{p}{p+n}$ then as it follows from the
results of \cite{Fateev:1996ea,Feigin:2001yq} for integer values of $p>2,$ three field model
after quantum group restriction describe rational unitarian CFT $M(N,p)$ with
central charge
\begin{equation}
c_{N,p}=\frac{3(p-2)}{N+p}\left(  \frac{p+N+2}{p(N+2)}\right)  \label{CC}%
\end{equation}
which is $\frac{SU_{2}(N)\otimes SU(p-2)}{SU(N+p-2)}$ coset CFT.\footnote{For
non-integer but rational values of parameter $p$ the model is rational but
non-unitarian.} The conformal dimensions of the basic fields $\phi_{n,m}^{k}$
in this model are characterized by three integers $k\leq N,m<p$ and $n<p+N,$
with $k=n-m$ $(\operatorname{mod}N)$ and are
\begin{equation}
\Delta_{m,n}=\frac{k(N-k)}{2N(N+2)}+\frac{(pn-(p+N)m)^{2}}{4pN(p+N)}\label{CD}%
\end{equation}
The fields $\phi_{n,m}^{k}$ up to normalization calculated in \cite{Baseilhac:1998eq} factor can
be represented as $\phi_{n,m}^{k}\rightarrow\sigma_{k}e^{-2i\omega_{mn}\varphi
}$ with $2\omega_{mn}=b(n-1)-(m-1)/b.$ The screening charges that are used to
calculate the correlation functions in the $M(N,p)$ theory have a form \cite{Kastor:1987dk}
\[
Q_{1}=\int d^{2}x\Psi\overline{\Psi}e^{i2b\varphi};\quad Q_{2}=\int d^{2}%
x\Psi^{+}\overline{\Psi}^{+}e^{-i2\varphi/b}.
\]
It means that correlation functions in this theory can be derived from the
correlation functions of the degenerate fields in parafermionic LFT by substitution
$b\rightarrow ib.$ In particular the structure constants in the $M(n,p)$
theory can be expressed in terms of integrals calculated in Section 3.

It was shown in ref \cite{Jimbo-Miwa} that the
CFT models $M(N,p)$ describe the critical behavior of generalized RSOS
statistical models. The spin variable $s$ in these models takes the integer
values with the restriction $|s-s^{\prime}|\leq N$ for the neighboring sites
of lattice. These models can be considered at the random (dynamical) lattice.
It is natural to think that the critical behavior of these \textquotedblleft
random\textquotedblright\ models is described by CFT $M(N,p)$ coupled with
parafermionic LFT. Solving parafermionic LFT on a  sphere we will be able to find KPZ
critical exponents and calculate the correlation functions in the coupled
theory. To derive more interesting quantities as the partition function on the
fluctuating disk we should solve boundary parafermionic LFT. One of the steps in this
direction is the calculation of bulk-boundary parafermionic Selberg integrals
which generalize the integrals calculated in \cite{Fateev:2007wk} for the Liouville boundary CFT.
\end{enumerate}
\section*{Acknowledgments}
This work was supported, in part, by cooperative CNRS-RFBR  grant PICS-09-02-93064. Work of A.~L. was supported  by DOE grant DE-FG02-96ER40949,  by RFBR  initiative interdisciplinary project grant 09-02-12446-OFI-m  and by Russian Ministry of Science and Technology under the Scientific Schools grant 3472.2008.2. The research of M.~B. and A.~L. was held within the framework of the Federal programs ``Scientific and Scientific-Pedagogical Personnel of Innovational Russia'' on 2009-2013 (state contracts No. P1339 and No. 02.740.11.5165).
\Appendix
\section{Parafermionic CFT (review)}\label{PF}
In this section we collect all needed facts about parafermionic CFT \cite{Fateev:1985mm}. In particular, we give an evidence of the statement that correlation function in parafermionic CFT containing at most three order fields and arbitrary number of parafermionic currents  is a polynomial (see \eqref{P-as-corr}) and derive the set of properties \eqref{prop} defining this polynomial uniquely up to a factor. We also explain our choice of normalization in \eqref{P-PP}. In the end of the appendix we collect some formulae for  the parafermionic polynomial \eqref{P-as-corr} as well as for the integral \eqref{Selberg-integral} in the dual case when parameters $k_{1}$, $k_{2}$ and $k_{3}$ are related by the condition \eqref{k-def-2}.

The $Z_N$ parafermionic algebra is generated by the holomorphic currents $\Psi_k(z)$, $k=1,\dots,N-1$ of fractional spins
\begin{equation}\label{D-k}
     \Delta_k=\frac{k(N-k)}{N},\qquad \bar{\Delta}_{k}=0.
\end{equation}
The OPE of these currents can be written in a form
\begin{equation}\label{OPE}
  \begin{aligned}
   &\Psi_k(z)\Psi_{k'}(0)=C_{k,k'}\,z^{-\frac{2kk'}{N}}\left(\Psi_{k+k'}(0)+O(z)\right),
  \quad\text{if}\quad k+k'<N,\\
  &\Psi_k(z)\Psi_{k'}(0)=C_{N-k,N-k'}\,z^{-\frac{2(N-k)(N-k')}{N}}\left(\Psi_{k+k'-N}(0)+O(z)\right),
  \quad\text{if}\quad k+k'>N,\\
  &\Psi_k(z)\Psi_{N-k}(0)=z^{-\frac{2k(N-k)}{N}}\left(1+\frac{2\Delta_k}{c}z^2\,T(0)+O(z^3)\right),
  \end{aligned}
\end{equation}
where
\begin{equation}\label{str-const}
    C_{k,k'}=
    \left[\frac{\Gamma(1+k+k')\Gamma(1+N-k)\Gamma(1+N-k')}{\Gamma(1+k)\Gamma(1+k')\Gamma(1+N-k-k')\Gamma(1+N)}\right]^{\frac{1}{2}}.
\end{equation}
Third line in \eqref{OPE} defines stress-energy tensor $T(z)$ which generates the conformal algebra of the theory with the central charge
\begin{equation} 
    c=\frac{2(N-1)}{(N+2)}.
\end{equation}
There is a set of antiholomorphic currents $\bar{\Psi}_k(\bar{z})$ with 
\begin{equation}
     \Delta_k=0,\qquad \bar{\Delta}_{k}=\frac{k(N-k)}{N}.
\end{equation}
forming operator algebra similar to \eqref{OPE}. Together they form the whole symmetry algebra of $Z_N$ parafermionic theory and all other fields belong to some representations of this algebra. We note that the algebra \eqref{OPE} is generated by the first parafermionic current $\Psi_1(z)$ which we denote for simplicity as 
\begin{equation}
    \Psi(z)\overset{\text{def}}{=}\Psi_{1}(z).    
\end{equation}
Analogously we define $\bar{\Psi}(\bar{z})$. It can be shown \cite{Fateev:1985mm} that any field in $Z_N$ theory can be obtained in the OPE's of the first parafermionic currents $\Psi(z)$ and $\bar{\Psi}(\bar{z})$ with one of the primary fields (it is convenient to include $\sigma_{0}(z,\bar{z})=\sigma_{N}(z,\bar{z})=1$.)
\begin{equation}
    \sigma_k(z,\bar{z}),\quad k=0,\dots, N-1.
\end{equation}
Fields $\sigma_k(z,\bar{z})$ are called by analogy with Ising model (which coincides with $Z_2$ theory) order fields and have conformal dimensions
\begin{equation}\label{d-k}
     d_k=\bar{d}_{k}=\frac{k(N-k)}{2N(N+2)}.
\end{equation}
We note that \eqref{D-k}, \eqref{OPE} and \eqref{d-k} have obvious symmetry with respect to $k\rightarrow N-k$. Sometimes it is convenient to define
\begin{equation}
      \Psi_k^{+}(z)\overset{\text{def}}{=}\Psi_{N-k}(z)\quad\text{and}\quad
      \sigma_k^{+}(z,\bar{z})\overset{\text{def}}{=}\sigma_{N-k}(z,\bar{z}).
\end{equation}
Parafermionic theories enjoy Kramers-Wannier duality, so it is also reasonable to define disorder fields 
\begin{equation}
      \mu_k(z,\bar{z}),\quad
      \mu_k^{+}(z,\bar{z})\overset{\text{def}}{=}\mu_{N-k}(z,\bar{z}),\quad
      k=0,\dots, N-1, 
\end{equation}
which have exactly the same conformal dimension as order fields \eqref{d-k}. It is convenient also to classify fields in $Z_N$ theory with respect to  both holomorphic and antiholomorphic Virasoro algebras\footnote{We note that the field primary with respect to Virasoro algebra is not necessary primary with respect to parafermionic algebra \eqref{OPE}.}
\begin{equation}\label{Phi-k-qq}
       \Phi^{k}_{(q,\bar{q})},\quad q,\bar{q}=k\mod 2\mathbb{Z},
\end{equation}
where the parameters $q$ and $\bar{q}$ are defined modulo $2N$ and
\begin{equation}\label{Phi-reflections}
     \Phi^{k}_{(q,\bar{q})}=\Phi^{k+N}_{(q+N,\bar{q}+N)}=\Phi^{N-k}_{(q-N,\bar{q}-N)}.
\end{equation}
The field $\Phi^{k}_{(k,k)}$ represents the order field $\sigma_k$, field $\Phi^{k}_{(-k,k)}$ corresponds to the disorder field $\mu_k$, $\Phi^{0}_{(2k,0)}$ corresponds to the parafermionic current $\Psi_k$ while $\Phi^{0}_{(0,2k)}$ corresponds to $\bar{\Psi}_k$  (we identify $\Psi_0=\Psi_N$ with identity operator). The system of the fields $\Phi^{k}_{(q,k)}$ obtained by the application of the holomorphic part of parafermionic algebra to one of the order fields can be conveniently drawn as shown on fig \ref{pic1}. Parafermionic currents $\Psi_k$ are drawn by big dots, order fields $\sigma_k$ by circles and disorder fields $\mu_k$ by crosses. 
\begin{figure}
\psfrag{o}{$0$}
\psfrag{k}{$k$}
\psfrag{q}{$q$}
	\centering
	\includegraphics[width=.8\textwidth]{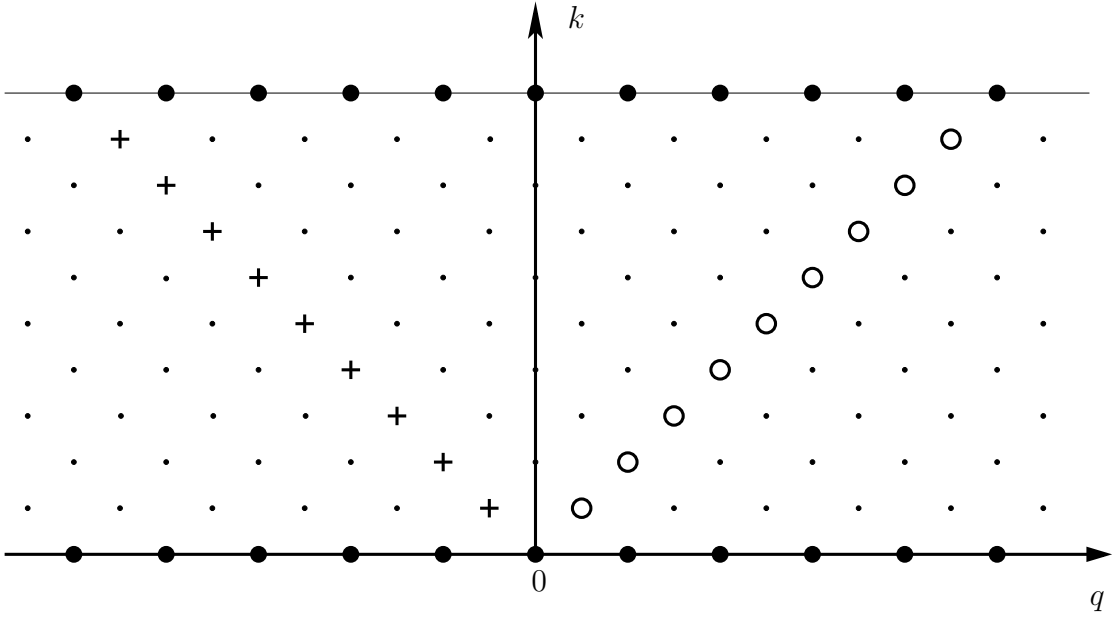}
        \caption{Digram describing the system of fields $\Phi^{k}_{(q,k)}$ in the $Z_{10}$ model.}\label{pic1}
\end{figure}
All other fields $\Phi^{k}_{(k+2l,k)}$ with $l=(-2k,\dots, N-k-1)$  shown by small dots on fig \ref{pic1} are parafermionic descendants of $\sigma_k$ or $\mu_k$ with conformal dimensions
\begin{equation}
        \Delta(\Phi^k_{(k+2l,k)})=d_k+\frac{l(N-k-l)}{N},\qquad
        \Delta(\Phi^k_{(k-2l,k)})=d_k+\frac{l(k-l)}{N},
\end{equation}
where $d_k$ is conformal dimension \eqref{d-k} of the order field $\sigma_k$. We note that parafermionic current $\Psi(z)$ acts horizontally from the left to the right on fig \ref{pic1}. The fields local with respect to themselves and with order fields are $\Phi_{(q,q)}^{k}$\footnote{Analogously the fields $\Phi_{(q,-q)}^{k}$ are local with themselves and disorder fields $\mu_{k}$.}

One can consider the multipoint correlation function
\begin{equation}\label{sss-ppp-corr-function}
     \mathfrak{G}_m^{(N|\,k_1,k_2,k_3)}(x_1,x_2,x_3|t_1,\dots,t_{n})=
     \langle\sigma_{k_1}(x_1,\bar{x}_1)\sigma_{k_2}(x_2,\bar{x}_2)\sigma_{k_3}(x_3,\bar{x}_3)
     \Psi(t_1)\bar{\Psi}(\bar{t}_{1})\dots\Psi(t_{n})\bar{\Psi}(\bar{t}_{n})\rangle,
\end{equation}
where $n=mN-k$. It is clear that the sufficient condition that this function is non-zero follows from the similar condition for the four-point function
\begin{equation}\label{sss-psi-4point}
     \mathfrak{g}_k^{(N|\,k_1,k_2,k_3)}(x_1,x_2,x_3|t)=
     \langle\sigma_{k_1}(x_1,\bar{x}_1)\sigma_{k_2}(x_2,\bar{x}_2)\sigma_{k_3}(x_3,\bar{x}_3)
     \Psi_{N-k}(t)\bar{\Psi}_{N-k}(\bar{t})\rangle. 
\end{equation}
In order to find out whether correlation function \eqref{sss-psi-4point} is non-zero we consider the operator product expansion \cite{Fateev:1985mm}
\begin{equation}\label{ss-OPE}
      \Phi_{(q_1,\bar{q}_{1})}^{k_1}\Phi_{(q_2,\bar{q}_{2})}^{k_2}=
      \begin{cases}
          \sum_{l=0}^{s}\left[\Phi^{k_1+k_2-2l}_{(q_1+q_2,\,\bar{q}_{1}+\bar{q}_{2})}\right],\quad
       s=\min(k_1,\,k_2)\quad\text{for}\quad k_1+k_2\leq N,\vspace{5mm}\\
      \sum_{l=0}^{s}\left[\Phi^{k_1+k_2-N+2l}_{(q_1+q_2-N,\,\bar{q}_{1}+\bar{q}_{2}-N)}\right],
       \quad s=\min(N-k_1,\,N-k_2)\quad\text{for}\quad k_1+k_2>N.           
      \end{cases}
\end{equation}
Evidently \eqref{sss-psi-4point} is non-zero if and only if in the OPE $\sigma_1\sigma_2\sigma_3$ there is a field $\Psi_k\bar{\Psi}_{k}$. Analyzing \eqref{ss-OPE} one can arrive that this is possible if one of the following conditions
\begin{subequations}
  \begin{align}\label{cond1}
      &k_1+k_{2}+k_{3}=2k,\quad0\leq k_{j}\leq k\leq N,\\\label{cond2}
      &k_1+k_{2}+k_{3}=N+2k,\quad0\leq k \leq k_{j}\leq N,
  \end{align}
\end{subequations}
is satisfied.  We consider in details the condition \eqref{cond1}. In this case correlation function \eqref{sss-ppp-corr-function} can be written in a form
\begin{equation}\label{G-P-1}
     \mathfrak{G}_m^{(N|\,k_1,k_2,k_3)}(x_1,x_2,x_3|t_1,\dots,t_{n})=
     \frac{\Xi(k_{1},k_{2},k_{3})\,|\mathbf{P}_{m}^{(N|\,k_1,k_2,k_3)}(x_1,x_2,x_3|t_1,\dots,t_{n})|^{2}}
     {|x_{12}|^{2\delta_{12}}|x_{13}|^{2\delta_{13}}|x_{23}|^{2\delta_{23}}\prod_{i,j}|t_{i}-x_{j}|^{\frac{2k_{j}}{N}}\prod_{i<j}|t_{ij}|^{\frac{4}{N}}},
\end{equation}
where $x_{ij}=x_{i}-x_{j}$, $t_{ij}=t_{i}-t_{j}$ and
\begin{equation*}\label{delta-ij}
     \delta_{ij}=\frac{k_{i}\,k_{j}}{N(N+2)}-\frac{(N-k+1)(k_{i}+k_{j}-k)}{(N+2)}.
\end{equation*}
The normalization factor $\Xi(k_{1},k_{2},k_{3})$ in \eqref{G-P-1}  will be defined below. Clearly  $\mathbf{P}(x_i|t_j)$  defined by \eqref{G-P-1} is a polynomial symmetric in $t_{1},\dots,t_{n}$. Projective invariance of correlation function \eqref{G-P-1} requires that this polynomial satisfy the set of properties \eqref{prop-1}--\eqref{prop-3}. The last two properties in \eqref{prop} follow immediately from the operator product expansions \eqref{OPE} and \eqref{ss-OPE}. Namely, correlation function \eqref{sss-ppp-corr-function} should respect OPE's
\begin{equation}\label{needed-OPE}
      \underbrace{\Psi\dots\Psi}_{N}\sim 1,\qquad
      \underbrace{\Psi\dots\Psi}_{N-k_{j}}\sigma_{k_{j}}\sim\mu_{k_{j}}.
\end{equation}
From other side applying \eqref{needed-OPE} to \eqref{G-P-1} one arrives to unwanted poles in \eqref{G-P-1}. The condition that these poles are cancelled is equivalent to\footnote{Here sign $\sim$ means proportional up to some polynomial.}
\begin{equation}\label{needed-clustering-prop}
      \begin{aligned}
        &\mathbf{P}_{m}^{(N|\,k_1,k_2,k_3)}(x_1,x_2,x_3|\underbrace{x,\dots,x}_{N},t_1,\dots,t_{n-N})\sim\prod_{i=1}^{n-N}(t_{i}-x)^{2},\\
        &\mathbf{P}_{m}^{(N|\,k_1,k_2,k_3)}(x_1,x_2,x_3|\underbrace{x_{j},\dots,x_{j}}_{N-k_{j}},t_1,\dots,t_{n+k_{j}-N})\sim
        \prod_{i=1}^{n+k_{j}-N}(t_{i}-x_{j}),\quad\text{for}\quad j=1,2,3,
      \end{aligned}
\end{equation} 
Clearly eq \eqref{needed-clustering-prop} implies \eqref{prop-4} and \eqref{prop-5}. One can see that there are no more obstructions and correlation function \eqref{G-P-1} with appropriate polynomial satisfy all eligibility requirements. Moreover, as it will be proven in appendix \ref{conjecture} the set of properties \eqref{prop} determines the polynomial $\mathbf{P}(x_i|t_j)$ in \eqref{G-P-1} uniquely up to factor which can be defined as follows. In order to obtain the four-point function \eqref{sss-psi-4point} from the multipoint one \eqref{sss-ppp-corr-function} one has to fuse $n=mN-k$ points $t_{1},\dots, t_{n}$ into $(m-1)$ clusters  with $N$ points and one cluster with $N-k$ points taking into account the structure constants of parafermionic algebra \eqref{str-const}, i.e.
\begin{equation}\label{correct-fusion}
     \underbrace{\underbrace{\Psi\dots\Psi}_{N}\;\dots\; \underbrace{\Psi\dots\Psi}_{N}}_{m-1}\;\;
      \underbrace{\Psi\dots\Psi}_{N-k}=\left[\frac{N!}{N^{\frac{N}{2}}}\right]^{m-1}
      \left[\frac{N!(N-k)!}{N^{N-k}k!}\right]^{\frac{1}{2}}\,\Psi_{N-k}.
\end{equation}
In our conventions we defined normalization of the polynomial $\mathbf{P}_{m}^{(N|\,k_1,k_2,k_3)}(x_1,x_2,x_3|t_1,\dots,t_{n})$ in \eqref{Polynomial-definition} according to \eqref{needed-clustering-prop} and \eqref{correct-fusion} as
\begin{multline}\label{P-normalization}
        \mathbf{P}_{m}^{(N|\,k_1,k_2,k_3)}(x_1,x_2,x_3|\underbrace{y_{1},\dots,y_{1}}_{N},\dots,\underbrace{y_{m-1},\dots,y_{m-1}}_{N},
        \underbrace{t,\dots,t}_{N-k})
        =\left[\frac{N!}{N^{\frac{N}{2}}}\right]^{m-1}\left[\frac{N!(N-k)!}{N^{N-k}k!}\right]^{\frac{1}{2}}\times\\\times
        \prod_{i<j=1}^{(m-1)}(y_{i}-y_{j})^{2N}\,
        \prod_{j=1}^{m-1}(y_{j}-x_{1})^{k_{1}}(y_{j}-x_{2})^{k_{2}}(y_{j}-x_{3})^{k_{3}}(y_{j}-t)^{2(N-k)}.
\end{multline}
According to \eqref{correct-fusion} and \eqref{P-normalization} the normalization constant $\Xi(k_{1},k_{2},k_{3})$  in \eqref{G-P-1} is defined through the  four-point correlation function \eqref{sss-psi-4point} as
\begin{equation}
     \mathfrak{g}_k^{(N|\,k_1,k_2,k_3)}(x_1,x_2,x_3|t)=
           \frac{\Xi(k_{1},k_{2},k_{3})}{|x_{12}|^{2\delta_{12}}|x_{13}|^{2\delta_{13}}|x_{23}|^{2\delta_{23}}\prod_{j}|t-x_{j}|^{\frac{2(N-k)k_{j}}{N}}}.
\end{equation}
From operator product expansion \eqref{ss-OPE} one immediately concludes that\footnote{In the second line in \eqref{Xi-as-product} we used the relations \eqref{Phi-reflections}.}
\begin{multline}\label{Xi-as-product}
    \Xi(k_{1},k_{2},k_{3})=\langle\sigma_{k_{1}}\sigma_{k_{2}}\Phi^{N-k_{3}}_{(N-k_{1}-k_{2},\,N-k_{1}-k_{2})}\rangle
    \langle\Phi^{k_{3}}_{(k_{1}+k_{2},\,k_{1}+k_{2})}\sigma_{k_{3}}\Psi_{N-k}\bar{\Psi}_{N-k}\rangle
    =\\=
    \langle\Phi^{k_{1}}_{(k_{1},\,k_{1})}\Phi^{N-k_{2}}_{(k_{2}-N,\,k_{2}-N)}\Phi^{N-k_{3}}_{(N-k_{1}-k_{2},\,N-k_{1}-k_{2})}\rangle
    \langle\Phi^{N-k_{3}}_{(k_{1}+k_{2}-N,\,k_{1}+k_{2}-N)}\Phi^{k_{3}}_{(k_{3},\,k_{3})}\Phi^{N}_{(N-2k,\,N-2k)}\rangle.
\end{multline}
An explicit expression for the three-point functions in $Z_{N}$ parafermionic theory was obtained in \cite{Fateev:1985mm}
\begin{equation}\label{3point-parafermions}
   \langle\Phi^{2j_{1}}_{(2m_{1},\,2\bar{m}_{1})}\Phi^{2j_{2}}_{(2m_{2},\,2\bar{m}_{2})}\Phi^{2j_{3}}_{(2m_{3},\,2\bar{m}_{3})}\rangle
   \simeq
   \genfrac{[}{]}{0pt}{}{j_{1}\;\;\,j_{2}\;\;\,j_{3}}{m_{1}\,m_{2}\,m_{3}}\genfrac{[}{]}{0pt}{}{j_{1}\;\;\,j_{2}\;\;\,j_{3}}{\bar{m}_{1}\,\bar{m}_{2}\,\bar{m}_{3}}\,
   \mathfrak{r}(j_{1},j_{2},j_{3}),
\end{equation}
where first two factors are $3j$-symbols and
\begin{multline}
   \mathfrak{r}(j_{1},j_{2},j_{3})=
   \left[\gamma\left(\frac{1}{N+2}\right)
   \gamma\left(\frac{N-2j_{1}+1}{N+2}\right)\gamma\left(\frac{N-2j_{2}+1}{N+2}\right)\gamma\left(\frac{N-2j_{3}+1}{N+2}\right)\right]^{\frac{1}{2}}  
   \times\\\times
   \frac{\Pi(j_{1}+j_{2}+j_{3}+1)\Pi(j_{2}+j_{3}-j_{1})\Pi(j_{1}+j_{3}-j_{2})\Pi(j_{1}+j_{2}-j_{3})}{\Pi(2j_{1})\Pi(2j_{2})\Pi(2j_{3})},
\end{multline}
with
\begin{equation*}
    \Pi(k)=k!\,\prod_{j=1}^{k}\gamma\left(\frac{j}{N+2}\right).
\end{equation*}
In \eqref{3point-parafermions} sign $\simeq$ means that the l.h.s. is not necessary zero if the r.h.s. does.  One has to use reflection properties 
\eqref{Phi-reflections} for the fields $\Phi^{2j}_{(2m,\,2\bar{m})}$ in \eqref{3point-parafermions} to ensure that the r.h.s. is non-zero (if it can not be done then this correlation function is indeed vanish). In our case we tuned parameters in \eqref{Xi-as-product} using this prescription and get
\begin{equation}
     \Xi(k_{1},k_{2},k_{3})=\frac{k!(N-k)!}{N!}\,\rho(k_{1},k_{2},k_{3}),
\end{equation}
where $\rho(k_{1},k_{2},k_{3})$ is given by \eqref{rho}. With explanations above we think that the choice of normalization in \eqref{P-PP} becomes clear\footnote{One has to take  a simultaneous look at eqs \eqref{G-P-1}, \eqref{P-as-corr} and \eqref{P-PP} to realize our choice of normalization.}.

In the end of this appendix we give an explicit expression for the parafermionic polynomial \eqref{P-as-corr} in the dual case when parameters $k_{1}$, $k_{2}$ and $k_{3}$ are related by
\begin{equation}\label{k-def-dual-2}
    k_1+k_2+k_3=N+2k, \quad 0\leq k\leq k_j\leq N.
\end{equation}
In this case polynomial defined by \eqref{P-as-corr} is expressed in terms of polynomial $\tilde{\mathbf{P}}_{m}^{(N|\,k_1,k_2,k_3)}(t_1,\dots,t_{n})$ as (here again $n=mN-k$)
\begin{equation}
     \mathcal{P}_{m}^{(N|\,k_1,k_2,k_3)}(t_1,\dots,t_{n})=
         (-1)^{m(N-k_{1})}\left[\frac{k!(N-k)!}{N!}\,\tilde{\rho}(k_{1},k_{2},k_{3})\right]^{\frac{1}{2}}\,
         \tilde{\mathbf{P}}_{m}^{(N|\,k_1,k_2,k_3)}(t_1,\dots,t_{n}),
\end{equation}
where $\tilde{\rho}(k_{1},k_{2},k_{3})$ is given by \eqref{rho} and $\tilde{\mathbf{P}}_{m}^{(N|\,k_1,k_2,k_3)}(t_1,\dots,t_{n})$ is defined similar to \eqref{pol-pol-def} by
\begin{equation}\label{pol-pol-def-dual}
   \tilde{\mathbf{P}}_{m}^{(N|\,k_1,k_2,k_3)}(t_1,\dots,t_{n})=(-1)^{m(N-k_1)}
   \lim_{x_3\rightarrow\infty}x_3^{k-mk_3}\,\tilde{\mathbf{P}}_{m}^{(N|\,k_1,k_2,k_3)}(0,1,x_3|t_1,\dots,t_{n}).
\end{equation}
Polynomial $\tilde{\mathbf{P}}_{m}^{(N|\,k_1,k_2,k_3)}(x_{1},x_{2},x_3|t_1,\dots,t_{n})$ satisfies the set of properties similar to \eqref{prop}
\begin{subequations}\label{prop-dual}
   \begin{align}\label{prop-1-dual} 
       &\tilde{\mathbf{P}}(\lambda x_i|\lambda t_{j})=
        \lambda^{m^2N-k}\tilde{\mathbf{P}}(x_i|t_{j}),\\\label{prop-2-dual}
       &\tilde{\mathbf{P}}\left(x_i^{-1}\bigl|t_j^{-1}\right)=
        (-1)^N\prod_{j=1}^3x_j^{-(mk_j-k)}\prod_{q=1}^{n}t_q^{-(2m-1)}\tilde{\mathbf{P}}(x_i|t_j),\\\label{prop-3-dual}
       &\tilde{\mathbf{P}}(x_i+\lambda|t_j+\lambda)=\tilde{\mathbf{P}}(x_i|t_j),\\[2mm]\label{prop-4-dual}
       &\tilde{\mathbf{P}}(x_1,x_2,x_3|\underbrace{x,\dots x}_{N+1},t_1,\dots t_{n-N-1})=0,\\\label{prop-5-dual}
       &\tilde{\mathbf{P}}(x_1,x_2,x_3|\underbrace{x_j,\dots x_j}_{N+1-k_j},t_1,\dots t_{n+k_j-N-1})=0\quad\text{for}\;j=1,2,3.
   \end{align}
\end{subequations}
One can shown that \eqref{prop-dual} determine the polynomial $\tilde{\mathbf{P}}_{m}^{(N|\,k_1,k_2,k_3)}(x_{1},x_{2},x_3|t_1,\dots,t_{n})$ unambiguously up to a total factor. It can be written similar to \eqref{Polynomial-definition}  as
\begin{equation}\label{Polynomial-definition-dual}
     \tilde{\mathbf{P}}_{m}^{(N|\,k_1,k_2,k_3)}(x_1,x_2,x_3|t_1,\dots,t_{n})=
      \Lambda_{m}^{-1}\,S[\tilde{p}(x_1,x_2,x_3|t_1,\dots,t_{n})],
\end{equation}
where $\Lambda_{m}$ is given by \eqref{Lambda} and the basic ``monomial'' $\tilde{p}(x_1,x_2,x_3|t_1,\dots,t_{n})$ is defined as follows.  We divide our points $t_1,\dots,t_n$ into $N$ groups
\begin{equation*}
    (t_1,t_2,\dots,t_n)\longrightarrow(s_1,\dots,s_N),
\end{equation*}
where  $N-k$ groups  contain $m$ different elements and the remaining $k$ groups  contain $m-1$ different elements. All groups $(s_1,\dots,s_N)$ are divided into four sets with $k_1-k$, $k_2-k$ and $k_3-k$ groups ($m$ elements in each group) and  $k$ groups ($(m-1)$ elements in each group)
\begin{equation*}
  \begin{aligned}
     &\chi_1=(s_{1},\dots,s_{k_1-k}),\quad&
     &\chi_2=(s_{k_1-k+1},\dots,s_{N-k_3}),\\
     &\chi_3=(s_{N-k_3+1},\dots,s_{N-k}),\quad&
     &\chi_4=(s_{N-k+1},\dots,s_{N}).
  \end{aligned}
\end{equation*} 
Then the basic monomial is
\begin{multline}\label{basic-monomial-dual}
     \tilde{p}(x_1,x_2,x_3|t_1,\dots,t_{n})=\\=
     \prod_{q=1}^N\prod_{i<j\in s_q}(t_i-t_j)^2\,
     \prod_{j\in\chi_1}(t_j-x_1)\prod_{j\in\chi_2}(t_j-x_2)\prod_{j\in\chi_3}(t_j-x_3)
      \prod_{j\in\chi_4}(t_j-x_1)(t_j-x_2)(t_j-x_3).
\end{multline}
We can define also the dual integral
\begin{equation}\label{Selberg-integral-dual}
     \tilde{\mathbf{J}}_{m}^{(N|\,k_1,k_2,k_3)}(A,B|g)=\int\limits_0^1\dots\int\limits_0^1
       \prod_{i=1}^{n}t_i^{A}(1-t_i)^{B}\prod_{i<j=1}^{n}|t_i-t_j|^{2g}\;
       \tilde{\mathbf{P}}_{m}^{(N|\,k_1,k_2,k_3)}(t_1,\dots,t_{n})
       dt_1\dots dt_{n}, 
\end{equation}
where the  integers $k_1$, $k_2$, $k_3$ and $k$ are restricted by the condition \eqref{k-def-dual-2}. It can be calculated using the method described in section \ref{integrals}. The result is
\begin{multline}\label{Int-g-rep-dual}
      \tilde{\mathbf{J}}_{m}^{(N|\,k_1,k_2,k_3)}(A,B|g)=C_{m}^{(N,k)}(g)\;
      \prod_{p=0}^{m-2}\frac{G_{k_1}^{(N)}(1+A+pNg+p)\,G_{k_2}^{(N)}(1+B+pNg+p)}
      {G_{N-k_3+k}^{(N)}\left(1+A+B+((m+p)N-k-1)g+m+p\right)}
      \times\\\times
      \frac{G_{k_1-k}^{(N-k)}(m+A+(m-1)Ng)\,G_{k_2-k}^{(N-k)}(m+B+(m-1)Ng)}
      {G_{N-k_3}^{(N-k)}(2m+A+B+((2m-1)N-k-1)g)},
\end{multline}
where function $G_{k}^{(N)}(x)$ is defined by \eqref{G-def} and function $C_{m}^{(N,k)}(g)$ is given by \eqref{C-k-answer}. 
\section{Proof of proposition \ref{main-proposition}}\label{conjecture}
In this appendix we give a proof of the following statement used in section \ref{integrals}. Namely, the polynomial $\mathbf{P}(x_{i}|t_{j})$ defined by the set of properties \eqref{prop} exists and unique up to a total factor. The existence was demonstrated in section \ref{integrals} by explicit expression, while the uniqueness will be proven below (we follow mainly \cite{Feigin:1993qr})\footnote{Along the same lines one can prove uniqueness of the dual polynomial defined by the set of properties \eqref{prop-dual}.}.

Let us consider the space $W$ of  symmetric polynomials $P(t_{1},\dots,t_{n})$ satisfying the following condition
\begin{equation}\label{diag-cond-1}
    P(\underbrace{x,\dots,x}_{N+1},t_{1},\dots,t_{n-N-1})=0.
\end{equation}
\begin{proposition}\label{P-Nx-proposition}
For any polynomial $P(t_{1},\dots,t_{n})\in W$
\begin{equation}\label{x-N-relat}
    P(\underbrace{x,\dots,x}_{N},t_{1},\dots,t_{n-N})\sim\prod_{j=1}^{n}(t_{j}-x)^{2},
\end{equation}
where sign $\sim$ means proportional up to some polynomial.
\end{proposition}
\begin{proof}\footnote{We note that there is a simple explanation of proposition \ref{P-Nx-proposition} from different point of view. Namely, one can consider dual space $W^{*}$ which is the space of all $n-$point correlation functions
\begin{equation*}
   P(t_{1},\dots,t_{n})=\langle\psi|e(t_{1})\dots e(t_{n})|0\rangle,
\end{equation*}
of the operator $e(t)=\sum_{j}e_{j}t^{-j-1}$, $[e_{i},e_{j}]=0$ satisfying $e(t)^{N+1}=0$, where $\langle\psi|$ runs over $W^{*}$. Clearly the relation $e(t)^{N+1}=0$   implies $e'(t)e(t)^{N}=0$. It follows from the relation $e(t)^{N+1}=0$ that $P(\underbrace{x,\dots,x}_{N+1},t_{1},\dots,t_{n-N-1})=0$,
while from $e'(t)e(t)^{N}=0$ we find that
$\frac{\partial}{\partial t_{j}}\,P(\underbrace{x,\dots,x}_{N},t_{1},\dots,t_{j},\dots,t_{n-N})\biggl|_{t_{j}=x}\hspace*{-10pt}=0$, which is equivalent to \eqref{x-N-relat}.
} 
Let us take in $P(u_{1},\dots,u_{N+1},t_{2},\dots,t_{n-N})$
\begin{equation}\label{u-lambda-rep}
   u_{j}=u+\lambda\xi_{j},\quad\text{with}\quad u=\frac{u_{1}+\dots+u_{N+1}}{N+1}\quad\text{and}\quad
   \sum_{j=1}^{N+1}\xi_{j}=0.
\end{equation}
We can expand at $\lambda\rightarrow0$
\begin{equation}\label{P-lambda-expansion}
     P(u+\lambda\xi_{1},\dots,u+\lambda\xi_{N+1},t_{2},\dots,t_{n-N})=P^{(0)}(u|t_{2},\dots,t_{n})+\lambda\,\mathbf{s}_{1}(\xi)P^{(1)}(u|t_{2},\dots,t_{n})+
     \dots,
\end{equation}
where $\mathbf{s}_{j}(\xi)$ are the elementary symmetric polynomials in $(\xi_{1},\dots,\xi_{N+1})$. Due to the conditions \eqref{diag-cond-1} and \eqref{u-lambda-rep} first two terms in the expansion \eqref{P-lambda-expansion} are actually vanish. So that the expansion \eqref{P-lambda-expansion} takes the form
\begin{equation}\label{P-lambda-expansion-2}
     P(u+\lambda\xi_{1},\dots,u+\lambda\xi_{N+1},t_{2},\dots,t_{n-N})=\lambda^{2}\,\mathbf{s}_{2}(\xi)P^{(2)}(u|t_{2},\dots,t_{n})+
     O(\lambda^{3}).
\end{equation}
Now we take $u_{1}=\dots=u_{N}=x$ and $u_{N+1}=t_{1}$ which is equivalent to $u=(Nx+t_{1})/(N+1)$ $\xi_{1}=\dots=\xi_{N}=(x-t_{1})/(\lambda(N+1))$ and $\xi_{N+1}=-N(x-t_{1})/(\lambda(N+1))$ and find that
\begin{equation*}
    P(\underbrace{x,\dots,x}_{N},t_{1},\dots,t_{n-N})\sim(t_{1}-x)^{2},
\end{equation*}
which due to the symmetry reasons implies \eqref{x-N-relat}.   
\end{proof}

Let us define the space $U_{\vec{Y}}$ labeled by Young tableau $\vec{Y}$ with $n$ cells
\begin{equation*}
    \vec{Y}=(p_{1},p_{2},\dots),\qquad \sum p_{j}=n,
\end{equation*}
as the space of polynomials satisfying
\begin{equation}\label{UY-def}
    P\in U_{\vec{Y}}\quad\text{if}\quad P(\vec{Y}')=0\quad\text{for any}\quad \vec{Y}'\geq\vec{Y},
\end{equation}
such that $\sum p'_{j}=n$.
In \eqref{UY-def} the notation $P(\vec{Y})$ means
\begin{equation}\label{P-on-Y}
    P(\vec{Y})\overset{\text{def}}{=}P(\underbrace{y_{1},\dots,y_{1}}_{p_{1}},\underbrace{y_{2},\dots,y_{2}}_{p_{2}},\dots),
\end{equation}
and by definition partition $\vec{Y}'$ is greater than partition $\vec{Y}$ if there exists $j$ such that $p_{i}'=p_{i}$ for $i<j$ and $p_{j}'>p_{j}$. This ordering is called lexicographic. It is clear that $W=U_{\vec{Y}_{0}}$
where $\vec{Y}_{0}$ is the following partition ($p_{1}=N+1$, $p_{2}=1,\dots$)
\begin{equation}\label{Y0}
    \begin{picture}(160,80)(0,70)
    \Thicklines
    \unitlength 1.5pt 
    \put(20,20){\line(0,1){30}}
    \put(20,60){\line(0,1){30}}
    \put(30,20){\line(0,1){30}}
    \put(30,60){\line(0,1){30}}
    
     \put(20,20){\line(1,0){10}}
     \put(20,30){\line(1,0){10}}
     \put(20,40){\line(1,0){10}}
     \put(21,55){\circle*{1}}
     \put(25,55){\circle*{1}}
     \put(29,55){\circle*{1}}
     \put(20,70){\line(1,0){10}}
     \put(20,80){\line(1,0){10}}
     \put(20,90){\line(1,0){10}}
     \put(30,20){\line(1,0){50}}
     \put(30,30){\line(1,0){50}}
     \put(90,20){\line(1,0){30}}
     \put(90,30){\line(1,0){30}}
     
     \put(40,20){\line(0,1){10}}
     \put(50,20){\line(0,1){10}}
     \put(60,20){\line(0,1){10}}
     \put(70,20){\line(0,1){10}}
     
     \put(100,20){\line(0,1){10}}
     \put(110,20){\line(0,1){10}}
     \put(120,20){\line(0,1){10}}
     
     \put(82,25){\circle*{1}}
     \put(85,25){\circle*{1}}
     \put(88,25){\circle*{1}}
    \put(60,5){\mbox{$n-N$}}
    \put(96,8){\vector(1,0){23}}
    \put(45,8){\vector(-1,0){23}}

    \put(-12,52){\mbox{$N+1$}}
    \put(0,65){\vector(0,1){23}}
    \put(0,43){\vector(0,-1){23}}
  \end{picture}
  \vspace*{2.5cm}
\end{equation}
It is also clear that $U_{\vec{Y}}\subseteq U_{\vec{Y}'}$ if $\vec{Y}<\vec{Y}'$. Thus we have a filtration
\begin{equation}
    \varnothing\subseteq\dots\subseteq U_{2}\subseteq U_{1}\subseteq U_{0},
\end{equation}
where we denote $U_{0}=W$, $U_{1}=U_{\vec{Y}_{1}}$ with $\vec{Y}_{1}$ being the partition with $n$ cells preceding the master partition $\vec{Y}_{0}$ in lexicographic order and so on. Let us define the quotient space
\begin{equation}
     \tilde{U}_{q}\overset{\text{def}}{=}U_{q-1}\backslash U_{q},
\end{equation}
which can be identified with the space of polynomials
\begin{equation}
    P(\vec{Y}_{q})\neq0\quad\text{and}\quad P(\vec{Y})=0\quad\text{for any}\quad \vec{Y}>\vec{Y_{q}}. 
\end{equation}
As for any polynomial $P\in\tilde{U}_{q}$ $P(\vec{Y}_{q})\neq0$ we can restrict it on the diagram $\vec{Y}_{q}=(p^{(q)}_{1},p^{(q)}_{2},\dots)$ and define according to \eqref{P-on-Y} the polynomial $\chi_{q}(y_{1},y_{2},\dots)$ as
\begin{equation}\label{P-chi-def}
   \chi_{q}(y_{1},y_{2},\dots)\overset{\text{def}}{=}P(\vec{Y}_{q}).
\end{equation}
It is clear that there is one to one correspondence between polynomials $P\in\tilde{U}_{q}$ and polynomials $\chi(y_{1},y_{2},\dots)$ defined by \eqref{P-chi-def}, i.~e. the map \eqref{P-chi-def} is injective. The polynomial $\chi(y_{1},y_{2},\dots)$ should vanish at the diagonal $y_{i}=y_{j}$ as
\begin{equation}
    \chi_{q}(y_{1},y_{2},\dots)\sim(y_{i}-y_{j})^{2p_{j}^{(q)}}\quad\text{for}\quad i<j.
\end{equation}
Indeed, for any polynomial $P\in\tilde{U}_{q}$ we can consider polynomial $\varrho(t_{1},\dots,t_{p_{i}^{(q)}+p_{j}^{(q)}})$\footnote{It depends also on variables $(y_{1},\dots,y_{i-1},y_{i+1},\dots,y_{j-1},y_{j+1},\dots)$, but this dependence is hidden for shortness.} 
\begin{multline}\label{P-partial-specification}
    \varrho(t_{1},\dots,t_{p_{i}^{(q)}+p_{j}^{(q)}})=
    P(\dots,\underbrace{y_{i-1},\dots,y_{i-1}}_{p_{i-1}^{(q)}},t_{1},\dots,t_{p_{i}^{(q)}}
    ,\underbrace{y_{i+1},\dots,y_{i+1}}_{p_{i+1}^{(q)}},\dots,\\
    \underbrace{y_{j-1},\dots,y_{j-1}}_{p_{j-1}^{(q)}},t_{p_{i}^{(q)}+1},\dots,t_{p_{i}^{(q)}+p_{j}^{(q)}}
    ,\underbrace{y_{j+1},\dots,y_{j+1}}_{p_{j+1}^{(q)}},\dots),
\end{multline}
i.e. we specify polynomial $P$ on almost whole diagram $\vec{Y}_{q}$ except $i$-th and $j$-th columns (here $i<j$ and hence $p_{i}^{(q)}\geq
p_{j}^{(q)}$). Clearly, the polynomial $\varrho(t_{1},\dots,t_{p_{i}^{(q)}+p_{j}^{(q)}})$ defined by \eqref{P-partial-specification}  is non-zero. Moreover, it should vanish if $p_{i}^{(q)}+1$ variables $t_{i}$ coincide, because in this case we get polynomial $P$ specified on a diagram which is definitely greater than $\vec{Y}_{q}$. Now as follows from proposition \ref{P-Nx-proposition}
\begin{equation*}
    \varrho(\underbrace{y_{i},\dots,y_{i}}_{p_{i}^{(q)}},t_{1},\dots,t_{p_{j}^{(q)}})\sim\prod_{a=1}^{\;p_{j}^{(q)}}(y_{i}-t_{a})^{2},
\end{equation*}
and hence
\begin{equation}\label{2-columns-prop}
    \varrho(\underbrace{y_{i},\dots,y_{i}}_{p_{i}^{(q)}},\underbrace{y_{j},\dots,y_{j}}_{p_{j}^{(q)}})\sim(y_{i}-y_{j})^{2p_{j}^{(q)}}.
\end{equation}
It follows from \eqref{2-columns-prop} that
\begin{equation}\label{chi-x-relat}
    \chi_{q}(y_{1},y_{2},\dots)\sim\prod_{i<j}(y_{i}-y_{j})^{2p_{j}^{(q)}}
\end{equation}

Now we restrict the space $W$ and define the space $W^{(k_{1},k_{2},k_{3})}\subseteq W$ as the space of polynomials  
$P\in W$ satisfying
\begin{equation}\label{additional-prop-app}
     P(\underbrace{x_{a},\dots,x_{a}}_{N+1-k_{a}},t_{1},\dots,t_{n+k_{a}-N-1})=0\quad\text{for}\quad 0\leq k_{a}\leq N\quad a=1,2,3.
\end{equation}
Where points $x_{1}$, $x_{2}$ and $x_{3}$ are some special points in a sense that we allow the coefficients of the polynomial $P$ depend on $x_{a}$.
In a same way we define the quotient spaces $\tilde{U}_{q}^{(k_{1},k_{2},k_{3})}\subseteq\tilde{U}_{q}$.
\begin{proposition}\label{P-Nxa-proposition}
For any polynomial $P(t_{1},\dots,t_{n})\in W^{(k_{1},k_{2},k_{3})}$
\begin{equation}\label{xa-N-relat}
    P(\underbrace{x,\dots,x}_{N-l},t_{1},\dots,t_{n-N+l})\sim\prod_{a=1}^{3}(x-x_{a})^{\varkappa_{a}},
\end{equation}
where
\begin{equation*}
   \varkappa_{a}=
       \begin{cases}
               k_{a}-l\quad\text{if}\quad l<k_{a},\\
               0\quad\text{if}\quad l\geq k_{a}.
        \end{cases}
\end{equation*}
\end{proposition}
\begin{proof}
Let us denote for simplicity   $f(t_{1},\dots,t_{N-l})\equiv P(t_{1},\dots,t_{N-l},t_{N-l+1},\dots)$.
Using the expansion valid for any symmetric function $f(t_{1},\dots,t_{N-l})$
\begin{multline*}
   f(x+\lambda\xi_{1},\dots,x+\lambda\xi_{N-l})=f(x,\dots,x)+\lambda\sum_{j}\xi_{j}\frac{\partial}{\partial t_{1}}
   f(t_{1},x,\dots,x)\Bigl|_{t_{1}=x}+\\+
   \lambda^{2}\sum_{j}\xi_{j}^{2}\frac{\partial^{2}}{\partial t_{1}^{2}}f(t_{1},x,\dots,x)\Bigl|_{t_{1}=x}+
   \lambda^{2}\sum_{i<j}\xi_{i}\xi_{j}\frac{\partial^{2}}{\partial t_{1}\partial t_{2}}f(t_{1},t_{2},x,\dots,x)\Bigl|_{t_{1}=t_{2}=x}+O(\lambda^{3}),
\end{multline*}
and taking $x=x_{a}$ we immediately conclude that due to \eqref{additional-prop-app}
\begin{equation}
    P(x_{a}+\lambda\xi_{1},\dots,x_{a}+\lambda\xi_{N-l},t_{1},\dots,t_{n-N+l})\sim\lambda^{\varkappa_{a}}
\end{equation}
taking $\xi_{j}=(x-x_{a})/\lambda$ we arrive at \eqref{xa-N-relat}.
\end{proof}
It follows from proposition \ref{P-Nxa-proposition} that the polynomial $\chi_{q}(y_{1},y_{2},\dots)$ specifies even more
\begin{equation}\label{chi-x-xa-relat}
    \chi_{q}(y_{1},y_{2},\dots)\sim\prod_{i}\prod_{a=1}^{3}(y_{i}-x_{a})^{\kappa_{i,a}^{(q)}}\,\prod_{i<j}(y_{i}-y_{j})^{2p_{j}^{(q)}},
\end{equation}
where
\begin{equation*}
   \kappa_{j,a}^{(q)}=
       \begin{cases}
               p_{j}^{(q)}+k_{a}-N\quad\text{if}\quad p_{j}^{(q)}>N-k_{a},\\
               0\quad\text{if}\quad p_{j}^{(q)}\leq N-k_{a}.
        \end{cases}
\end{equation*}
Now let us define according to \eqref{chi-x-xa-relat}
\begin{equation}\label{nu-q-definition}
    \nu_{j}^{(q)}=2(j-1)+\frac{1}{p_{j}^{(q)}}\left(2\sum_{i>j}p_{i}^{(q)}+\sum_{a=1}^{3}\kappa_{j,a}^{(q)}\right),
\end{equation}
which is the minimal possible degree of polynomial $\chi_{q}(y_{1},y_{2},\dots)$ in variable $y_{j}$ divided by $p_{j}^{(q)}$ (we note that $\nu_{j}^{(q)}$ is not necessary an integer number).
It is clear that any polynomial $P\in\tilde{U}_{q}^{(k_{1},k_{2},k_{3})}$ is at least of degree
\begin{equation}
   \nu^{(q)}=\lceil\max\nu_{j}^{(q)}\rceil
\end{equation}
in each variable $t_{1},\dots,t_{n}$ where $\lceil x\rceil$ means ceiling function i.e. $\lceil x\rceil=\min\{m\in\mathbb{Z}|m\geq x\}$. 

Now we turn to the special case of polynomials defined by the set of properties \eqref{prop}. In this case integers $k_{1}$, $k_{2}$ and $k_{3}$ are restricted by the condition \eqref{k-def-1} and $n=mN-k$. It follows from \eqref{prop-2} the polynomial $\mathbf{P}_{m}^{(N|\,k_{1},k_{2},k_{3})}(x_{1},x_{2},x_{3}|t_{1},\dots,t_{n})$ is of degree $2(m-1)$ in each variable $t_{j}$. We will show below that $2(m-1)$ is the minimal possible degree for the polynomial from $W^{(k_{1},k_{2},k_{3})}$ with the condition that integers $k$, $k_{1}$, $k_{2}$ and $k_{3}$ are related by \eqref{k-def-1} and $n=mN-k$ and if such a polynomial exists it should be unique.
Let us compute $\nu_{j}^{(1)}$ defined by \eqref{nu-q-definition}. By definition the diagram $\vec{Y}_{1}$ preceding in lexicographic ordering the master diagram $\vec{Y}_{0}$ drawn in \eqref{Y0}  is
\begin{equation}\label{Y1}
    \begin{picture}(180,100)(0,70)
    \Thicklines
    \unitlength 1.5pt 
    \put(20,20){\line(0,1){30}}
    \put(20,60){\line(0,1){15}}
    \put(20,85){\line(0,1){15}}
    \put(30,20){\line(0,1){30}}
    \put(30,60){\line(0,1){15}}
    \put(30,85){\line(0,1){15}}
    \put(40,85){\line(0,1){15}}
    \put(50,85){\line(0,1){15}}
    \put(60,85){\line(0,1){15}}
    \put(70,85){\line(0,1){15}}
    \put(100,85){\line(0,1){15}}
    \put(110,85){\line(0,1){15}}
    \put(120,85){\line(0,1){15}}
    
     \put(20,20){\line(1,0){10}}
     \put(20,30){\line(1,0){10}}
     \put(20,40){\line(1,0){10}}
     \put(21,55){\circle*{1}}
     \put(25,55){\circle*{1}}
     \put(29,55){\circle*{1}}
     \put(41,55){\circle*{1}}
     \put(45,55){\circle*{1}}
     \put(49,55){\circle*{1}}
     \put(61,55){\circle*{1}}
     \put(65,55){\circle*{1}}
     \put(69,55){\circle*{1}}
     \put(101,55){\circle*{1}}
     \put(105,55){\circle*{1}}
     \put(109,55){\circle*{1}}
     \put(121,55){\circle*{1}}
     \put(125,55){\circle*{1}}
     \put(129,55){\circle*{1}}
     \put(20,70){\line(1,0){10}}
     \put(20,90){\line(1,0){10}}
     \put(30,20){\line(1,0){50}}
     \put(30,30){\line(1,0){50}}
     \put(90,20){\line(1,0){40}}
     \put(90,30){\line(1,0){40}}
     
     \put(30,40){\line(1,0){50}}
     \put(30,70){\line(1,0){50}}
     \put(20,90){\line(1,0){60}}
     \put(20,100){\line(1,0){60}}
     \put(90,40){\line(1,0){40}}
     \put(90,70){\line(1,0){40}}
     \put(90,90){\line(1,0){30}}
     \put(90,100){\line(1,0){30}}
     
     \put(40,20){\line(0,1){30}}
     \put(50,20){\line(0,1){30}}
     \put(60,20){\line(0,1){30}}
     \put(70,20){\line(0,1){30}}
     
     \put(100,20){\line(0,1){30}}
     \put(110,20){\line(0,1){30}}
     \put(120,20){\line(0,1){30}}
     
     \put(40,60){\line(0,1){15}}
     \put(50,60){\line(0,1){15}}
     \put(60,60){\line(0,1){15}}
     \put(70,60){\line(0,1){15}}
     
     \put(100,60){\line(0,1){15}}
     \put(110,60){\line(0,1){15}}
     \put(120,60){\line(0,1){15}}
     \put(130,20){\line(0,1){30}}
     \put(130,60){\line(0,1){10}}
     
     \put(82,25){\circle*{1}}
     \put(85,25){\circle*{1}}
     \put(88,25){\circle*{1}}
     \put(82,35){\circle*{1}}
     \put(85,35){\circle*{1}}
     \put(88,35){\circle*{1}}
     \put(82,45){\circle*{1}}
     \put(85,45){\circle*{1}}
     \put(88,45){\circle*{1}}
     \put(82,75){\circle*{1}}
     \put(85,75){\circle*{1}}
     \put(88,75){\circle*{1}}
     \put(82,85){\circle*{1}}
     \put(85,85){\circle*{1}}
     \put(88,85){\circle*{1}}
     \put(82,95){\circle*{1}}
     \put(85,95){\circle*{1}}
     \put(88,95){\circle*{1}}
     
     \put(102,80){\circle*{1}}
     \put(105,80){\circle*{1}}
     \put(108,80){\circle*{1}}
     \put(112,80){\circle*{1}}
     \put(115,80){\circle*{1}}
     \put(118,80){\circle*{1}}
     
     \put(62,80){\circle*{1}}
     \put(65,80){\circle*{1}}
     \put(68,80){\circle*{1}}
     
     \put(42,80){\circle*{1}}
     \put(45,80){\circle*{1}}
     \put(48,80){\circle*{1}}
     
     \put(22,80){\circle*{1}}
     \put(25,80){\circle*{1}}
     \put(28,80){\circle*{1}}
    \put(72,6){\mbox{$m$}}
    \put(90,8){\vector(1,0){40}}
    \put(62,8){\vector(-1,0){40}}

    \put(-3,56){\mbox{$N$}}
    \put(0,71){\vector(0,1){27}}
    \put(0,47){\vector(0,-1){27}}
    
    \put(133,40){\mbox{$(N-k)$}}
    \put(145,55){\vector(0,1){13}}
    \put(145,33){\vector(0,-1){13}}
  \end{picture}
  \vspace*{2.5cm}
\end{equation}
Using definition \eqref{nu-q-definition} we find
\begin{equation}
   \nu_{j}^{(1)}=
                 \begin{cases}
                    2(m-1)+\frac{\sum k_{a}-2k}{N},\quad\text{for}\quad j=1,\dots,m-1\\
                    2(m-1)+\frac{\sum\kappa_{m,a}^{(1)}}{N-k},\quad\text{for}\quad j=m. 
                 \end{cases}
\end{equation}
Now, using the condition
\begin{equation*}
    \sum k_{a}=2k,\qquad 0\leq k_{a}\leq k\leq N,
\end{equation*}
we find that $\kappa_{m,a}=0$ and hence
\begin{equation}
    \nu^{(1)}=\nu^{(1)}_{1}=\dots=\nu^{(1)}_{m}=2(m-1).
\end{equation}
Moreover, one can easily see that\footnote{It is clear that for any diagram $\vec{Y}_{q}<\vec{Y}_{1}$ with the total number of columns greater than $m$ the degree $\nu^{(q)}>2(m-1)$ (it follows immediately from the definition \eqref{nu-q-definition}). So, we must consider only diagrams with $m$ columns. Let us take one of them then evidently $p_{m-1}^{(q)}=N-l$ and $p_{m}^{(q)}=N-k+l+\epsilon$ with some positive integer $l$ and non-negative integer $\epsilon$. From definition \eqref{nu-q-definition} we find that
\begin{equation*}
   \nu_{m-1}^{(q)}=2(m-1)+\frac{1}{(N-l)}\left(4l-2k+2\epsilon+\sum_{a=1}^{3}\kappa_{m-1,a}^{(q)}\right).
\end{equation*}
Using that $\sum_{a=1}^{3}\kappa_{m-1,a}^{(q)}\geq\sum_{a=1}^{3}(k_{a}-l)=2k-3l$ we find that the expression in brackets is greater than $l+2\epsilon$ and hence $\nu_{m-1}^{(q)}>2(m-1)$.} 
\begin{equation}\label{nu-equality}
   \nu^{(q)}>\nu^{(1)}\quad\text{for}\quad q>1.
\end{equation}
So, if there exists a polynomial $P\in W^{(k_{1},k_{2},k_{3})}$ with $\sum k_{a}=2k$ and $0\leq k_{a}\leq k\leq N$ of degree $2(m-1)$ in each variable 
$t_{1},\dots,t_{n}$ it should be unique, because in this case $\dim(\tilde{U}_{1}^{(k_{1},k_{2},k_{3})})\leq1$, while as follows from \eqref{nu-equality} $\dim(\tilde{U}_{q}^{(k_{1},k_{2},k_{3})})=0$ for $q>1$ and hence
\begin{equation}
  \dim (W^{(k_{1},k_{2},k_{3})})=\sum_{q}\dim(\tilde{U}_{q}^{(k_{1},k_{2},k_{3})})\leq1.
\end{equation}
An explicit form of this polynomial was given in section \ref{integrals}. Of course, the total normalization remains unfixed under above consideration. We chose it regarding our needs as in   \eqref{Polynomial-definition}.
\section{Proof of proposition \ref{Selberg-prop-1}.}\label{Selberg-prop-proof}
In this appendix we give a proof of proposition \ref{Selberg-prop-1}. For simplicity, we consider the case $k_{1}=k_{2}=k_{3}=0$. We define the polynomial $P(t_{1},\dots,t_{n})$ with $n=mN$  which is proportional up to a factor to  the polynomial $\mathbf{P}_{m}^{(N|0,0,0)}(t_{1},\dots,t_{n})$. It can be represented as  a symmetrization of the ``monomial'' $p(t_{1},\dots,t_{n})$
\begin{equation}
    p(t_{1},\dots,t_{n})=\prod_{q=1}^{N}\prod_{i<j\in s_{q}}(t_{i}-t_{j})^{2},
\end{equation}
where $(s_{1},\dots,s_{N})$ are the groups of $m$ different variables $t_{j}$. For non-negative integer values of the parameter $g$ we can define the polynomial
\begin{equation}
   D(t_{1},\dots,t_{n})=\prod_{i<j}|t_{i}-t_{j}|^{2g}P(t_{1},\dots,t_{n}),
\end{equation}
which can be expanded into the sum
\begin{equation}\label{polynom-expansion}
   D(t_{1},\dots,t_{n})=\sum_{\nu_{j}}c_{\nu_{1},\dots,\nu_{n}}t_{1}^{\nu_{1}}\dots t_{n}^{\nu_{n}}.
\end{equation}
The proposition \ref{Selberg-prop-1} in this case reads as
\begin{proposition}
For $\nu_{1}\leq\nu_{2}\leq\dots\leq\nu_{n}$ the non-zero values of   $c_{\nu_{1},\dots,\nu_{n}}$ in \eqref{polynom-expansion} occur when the integers $\nu_{pN+j}$ with $p=0,\dots,m-1$ and $j=1,\dots,N$ satisfy the inequalities 
\begin{equation}\label{polynom-ineq}
      p+(pN+j-1)g\leq\nu_{pN+j}\leq(m+p-1)+((m+p)N+j-2)g.
\end{equation}
\end{proposition}
\begin{proof}
Clearly, the polynomial \eqref{polynom-expansion} is a homogeneous symmetric polynomial of degree $$(m-1)n+n(n-1)g,$$ i.e. $\sum_{j=1}^{n}\nu_{j}=(m-1)n+n(n-1)g$. This summation formula together with the assumed ordering $\nu_{1}\leq\nu_{2}\leq\dots\leq\nu_{n}$ implies that 
\begin{equation}
 \nu_{n}\geq m-1+(n-1)g.
\end{equation}
To estimate $\nu_{n-k}$ for  we use the fact that
\begin{equation}\label{ineq-stage-1}
   \sum_{j=1}^{n-k}\nu_{j}\geq(m-1)n+n(n-1)g-\nu_{k}^{\text{max}},
\end{equation}
where $\nu_{k}^{\text{max}}$ is the maximal collective degree of the polynomial $D(t_{1},\dots,t_{n})$ in variables $t_{n-k+1},\dots,t_{n}$. In order to find $\nu_{k}^{\text{max}}$ we take
\begin{equation*}
   t_{n-k+1}=\tau\xi_{n-k+1},\quad t_{n-k+2}=\tau\xi_{n-k+2},\quad\dots\quad t_{n}=\tau\xi_{n},
\end{equation*}
and find the leading asymptotic of $D(t_{1},\dots,t_{n-k},\tau\xi_{n-k+1},\dots,\tau\xi_{n})$ in the limit $\tau\rightarrow\infty$.  First of all, we note that
\begin{equation*}
     \nu_{k}^{\text{max}}=\tilde{\nu}_{k}^{\text{max}}+\hat{\nu}_{k}^{\text{max}},
\end{equation*}
where the degree $\tilde{\nu}_{k}^{\text{max}}$ comes for the asymptotic of the polynomial $\prod_{i<j}|t_{i}-t_{j}|^{2g}$ while $\hat{\nu}_{k}^{\text{max}}$ comes from the polynomial $P(t_{1},\dots,t_{n})$. The first degree is a simple one
\begin{equation}\label{Deg-1}
    \tilde{\nu}_{k}^{\text{max}}=k(2n-k-1)g.
\end{equation}
In order to find $\hat{\nu}_{k}^{\text{max}}$ we note that the asymptotic of the polynomial  $P(t_{1},\dots,t_{n-k},\tau\xi_{n-k+1},\dots,\tau\xi_{n})$ in the limit $\tau\rightarrow\infty$ depends essentially on an integer part of the fraction $k/N$. It is convenient to represent
\begin{equation}\label{k-p-j}
    k=(m-p)N-j,\qquad p=0,\dots,m-1,\quad j=1,\dots,N,
\end{equation}
i.e. $n-k=pN+j$ (we remind that $n=mN$). It is clear that the leading asymptotic of the polynomial  $P(t_{1},\dots,t_{n-k},\tau\xi_{n-k+1},\dots,\tau\xi_{n})$ comes from the basic ``monomials'' such that $j$ groups contain $(m-p-1)$ variables $\tau\xi_{i}$ and $(N-j)$ groups contain $(m-p)$ variables. Using simple algebra one can find
\begin{equation}\label{Deg-2}
   \hat{\nu}_{k}^{\text{max}}=(m-p)(m+p-1)N-2pj.
\end{equation}
Using \eqref{ineq-stage-1}, \eqref{Deg-1}, \eqref{k-p-j} and \eqref{Deg-2} one can find
\begin{equation*}
   \sum_{i=1}^{pN+j}\nu_{i}\geq p(pN+j)+(pN+j)(pN+j-1)g-p(N-j),
\end{equation*}
and hence from the assumed ordering of the integers $\nu_{1}\dots\nu_{pN+j}$ we get
\begin{equation*}
       \nu_{pN+j}\geq p+(pN+j-1)g-\frac{p(N-j)}{pN+j}.
\end{equation*}
We note that the number $p(N-j)/(pN+j)$ is always smaller than $1$ and hence  it can be dropped in the inequality above.  So, we get a final formula
\begin{equation*}
    \nu_{pN+j}\geq p+(pN+j-1)g,
\end{equation*}
which is the left inequality in  \eqref{polynom-ineq}. Now, using the fact that
\begin{equation}\label{D-D}
     D(t_{1},\dots,t_{n})=\prod_{j=1}^{n}t_{j}^{2(m-1)+2(n-1)g}D(1/t_{1},\dots,1/t_{n}),
\end{equation}
we can expand the r.h.s. of \eqref{D-D} as in \eqref{polynom-expansion} with $\nu_{1}'\leq\nu_{2}'\leq\dots\leq\nu_{n}'$. We note that \eqref{D-D} implies that $\nu_{i}'=2(m-1)+2(n-1)g-\nu_{n+1-i}$ and hence the inequality $\nu_{pN+j}'\geq p+(pN+j-1)g$ implies that 
\begin{equation*}
  \nu_{pN+j}\leq(m+p-1)+((m+p)N+j-2)g,
\end{equation*}
which is the right inequality in \eqref{polynom-ineq}.
\end{proof}
In the case of general parameters $k_{j}$ and $k$ the proof of proposition \ref{Selberg-prop-1} can be performed in a same way. We hope that the explicit form of the inequalities \eqref{ineq} will help interested reader to repeat above arguments in this general case too.
\section{Sigma model representation and minisuperspace approximation for three exponential model}\label{MSP-approximation}
In this appendix we consider minisuperspace approximation for CFT (\ref{a}) and its sigma model representation. For this purpose we consider this theory in the dual representation with screening fields $J_{i}$ defined by \eqref{J} and \eqref{j3}. At first we consider this model with only one screening field $\ J_{1}\overline{J_{1}},$ which corresponds to the
parafermionic Liouville CFT. It follows from the explicit form of the current $J_{1}$ that the action
\begin{equation}
A_{1}=\int d^{2}x\left(  \frac{1}{16\pi}\sum_{i}^{3}(\partial_{\mu}\varphi
_{i})^{2}-N\mu J_{1}\overline{J_{1}}+\frac{R^{\left(  2\right)  }%
}{8\pi}\Phi\right)  \label{a1}%
\end{equation}
where our dilaton field is $\Phi=\frac{i}{2\alpha}\varphi_{3}+\frac{1}%
{2\alpha_{1}}\varphi_{1}+\frac{1}{2\alpha_{2}}\varphi_{2}$ can be written in
terms of sigma model with complex metric. In minisuperspace approximation when
$N$ and parameters $\alpha,\alpha_{i}$ are large we can use the relation:
$\alpha^{2}=\alpha_{1}^{2}+\alpha_{2}^{2}.$ To simplify the notations we
denote $\varphi_{1}=x_{1},\varphi_{2}=x_{2},\varphi_{3}=x_{0}.$ Then our
metric can be written in the form.%
\begin{align}
ds^{2} &  =(1+m\alpha^{2}e^{\frac{x_{2}}{2\alpha_{2}}})dx_{0}^{2}%
+(1-m\alpha_{1}^{2}e^{\frac{x_{2}}{2\alpha_{2}}})dx_{1}^{2}+(1-m\alpha_{2}%
^{2}e^{\frac{x_{2}}{2\alpha_{2}}})dx_{2}^{2}-\nonumber\\
&  -2i\alpha\alpha_{1}me^{\frac{x_{2}}{2\alpha_{2}}}dx_{0}dx_{1}%
+2i\alpha\alpha_{2}me^{\frac{x_{2}}{2\alpha_{2}}}dx_{0}dx_{2}+2\alpha
_{1}\alpha_{2}me^{\frac{x_{2}}{2\alpha_{2}}}dx_{1}dx_{2}\label{m}%
\end{align}
where parameter $m=4\pi\mu$. Our metric $g_{ik}$ is complex in
the \textquotedblleft Euclidean\textquotedblright\ space, however in
\textquotedblleft Minkowski\textquotedblright\ space: $x_{0}\rightarrow
ix_{0}$ this metric is real regular metric with signature $(-1,1,1).$ This
metric possesses the properties $g=\det[g_{ik}]=1,g^{ik}(m)=g_{ik}(-m).$ The
metric $g_{ik}$ is not conformally flat. It has non-zero Ricci tensor. The
scalar curvature for this metric $R=-\frac{1}{4}me^{\frac{x_{2}}{2\alpha_{2}}%
}.$ It is small for $x_{2}\rightarrow-\infty$ and becomes large for
$x_{2}\rightarrow\infty,$ where we expect strong interaction. To compensate
the Ricci tensor we should take into account dilaton field $\Phi=\frac
{i}{2\alpha}x_{0}+\frac{1}{2\alpha_{1}}x_{1}+\frac{1}{2\alpha_{2}}x_{2}.$ The
action for renormalization group flow is
\[
S=\int d^{3}x\sqrt{g}e^{-\Phi}(R+|\nabla\Phi|^{2})
\]
and the condition of conformal invariance has a form
\begin{equation}
R_{ik}+\nabla_{i}\nabla_{k}\Phi=0\label{rg}%
\end{equation}
It can be checked that this equation is satisfied with our linear dilaton
field $\Phi.$ The equation that describes zero modes dynamics in
minisuperspace approximation has a form
\begin{equation}
-\frac{e^{\Phi}}{\sqrt{g}}\partial_{i}e^{-\Phi}\sqrt{g}g^{ik}\partial
_{k}\Theta=\Delta\Theta\label{req}%
\end{equation}
We see that our metric depends only on the variable $x_{2},$ it follows from
the fact that field $J_{1}\overline{J_{1}}$ interacts essentially only with
field $\varphi_{2}.$ It means that function $\Theta$ depends non-trivially
only on variable $x_{2}.$ It is natural take this function in the form
\[
\Theta=\exp\{\Phi/2+iqx_{0}+ip_{1}x_{1}\}\Psi(x_{2}).
\]
Then eq (\ref{req}) can be reduced to the hypergeometric form. The unique
regular at $x_{2}\rightarrow\infty$ properly normalized solution to this
equation with eigenvalue
\begin{equation}
\Delta=-\frac{1}{16\alpha^{2}}+q^{2}+\frac{1}{16\alpha_{1}^{2}}+p_{1}%
^{2}+\frac{1}{16\alpha_{2}^{2}}+p_{2}^{2}\label{eig}%
\end{equation}
has a form
\begin{align}
\Psi & =\frac{\Gamma(\frac{1}{2}-2i\alpha_{2}p_{2}+2\alpha q-2i\alpha_{1}%
p_{1})\Gamma(\frac{1}{2}-2i\alpha_{2}p_{2}-2\alpha q+2i\alpha_{1}p_{1}%
)}{\Gamma(-4i\alpha_{2}p_{2})}(m\alpha_{2}^{2})^{-2i\alpha_{2}p_{2}+\frac
{1}{2}}\nonumber\\
& \times\exp\left\{  -\frac{x_{2}}{4\alpha_{2}}+x_{2}(\frac{i\alpha_{1}p_{1}%
}{\alpha_{2}}-\frac{\alpha q}{\alpha_{2}})\right\}  F\left(  A,B,1,-\frac
{1}{m\alpha_{2}^{2}}e^{-\frac{x_{2}}{2\alpha_{2}}}\right)  \label{sol}%
\end{align}
where $A=\frac{1}{2}+2i\alpha_{2}p_{2}+2\alpha q-2i\alpha_{1}p_{1},B=\frac
{1}{2}-2i\alpha_{2}p_{2}+2\alpha q-2i\alpha_{1}p_{1}$ and numerical factor in
eq (\ref{sol}) is chosen in such way that asymptotic of function $\Psi(x_{2})$
in the free region $x_{2}\rightarrow\infty$ has a canonical form
\begin{equation}
\Psi(x_{2})=e^{ip_{2}x_{2}}+S_{2}(p_{1},p_{2},q)e^{-ip_{2}x_{2}}\label{cf}%
\end{equation}
It follows from explicit form of our solution that
\begin{equation}
S_{2}(p_{1},p_{2},q)=(m\alpha_{2}^{2})^{-4i\alpha_{2}p_{2}}\frac
{\Gamma(4i\alpha_{2}p_{2})\gamma(\frac{1}{2}-2i\alpha_{2}p_{2}+2\alpha
q-2i\alpha_{1}p_{1})}{\Gamma(-4i\alpha_{2}p_{2})\gamma(\frac{1}{2}%
+2i\alpha_{2}p_{2}+2\alpha q-2i\alpha_{1}p_{1})}\label{s}%
\end{equation}
If we take into account that $m=4\pi\mu$ and that as follows from
eq (\ref{ba})
\begin{equation}
\left(\frac{\pi M}{4\alpha_{2}^{2}}\right)=
\left(  \pi\mu\gamma\left(  \frac{1}{4\alpha_{2}^{2}}\right)  \right)
^{2\alpha_{2}^{2}}\approx(4\pi\mu\alpha_{2}^{2})^{2\alpha_{2}^{2}}\label{mmu}%
\end{equation}
we see from eqs (\ref{RRR},\ref{RL}) that in the minisuperspace limit where
$\alpha_{i}>>1$ and 
$$2\alpha_{j}a_{j}=\frac{1}{2}+i2\alpha_{j}p_{j},\quad
2\alpha a_{3}=\frac{1}{2}+2\alpha q$$ 
are finite $S_{2}(p_{1},p_{2},q)$ coincides with reflection amplitude $R_{2}(\vec{a}).$

In the same way if we substitute in the action (\ref{a1}) the term $\ \mu
J_{1}\overline{J_{1}}$ by the term $\widetilde{\mu}J_{2}\overline{J_{2}}$ and
in the eq (\ref{mmu}) $\mu\rightarrow\widetilde{\mu}$ and $\alpha
_{2}\rightarrow\alpha_{1}$ we can calculate coefficient $S_{1}(p_{1},p_{2},q)$
for outcoming wave $e^{-ip_{1}x_{1}}$ and  check that it coincides with
minisuperspace limit of reflection amplitude $R_{1}(\vec{a}).$

To derive the minisuperspace limit of the reflection amplitude $R_{1,2}(\vec{a})$
we should take in the action both terms $\ \mu J_{1}\overline{J_{1}%
}+\widetilde{\mu}J_{2}\overline{J_{2}}.$ However, in this way we have some
problem with the correct definition of the action. The reason is that the operator
product expansion of currents $J_{1}$ and $J_{2}$ is rather singular
\begin{equation}
J_{1}(z)J_{2}(z^{\prime})=(z-z^{\prime})^{-2}\exp\left(  \frac{\varphi_{2}%
(z)}{2\alpha_{2}}+\frac{\varphi_{1}(z)}{2\alpha_{1}}\right)  +... \label{sing}%
\end{equation}
and we should renormalize our action by adding counterterms in such a way that
the symmetry of our action becomes not broken. For this we note that the primary field
with dimension $\frac{N+4}{N+2}$
\begin{equation}
G(z)=(i\alpha\partial_{z}\varphi_{3}+\alpha_{1}\partial_{z}\varphi_{1}%
+\alpha_{2}\partial_{z}\varphi_{2})\exp\left(  -\frac{i\varphi_{3}(z)}%
{2\alpha}\right)  \label{G}%
\end{equation}
commutes with all screening charges $\mathbf{S}_{++}^{+},\mathbf{S}_{+-}^{-}$
and $\mathbf{S}_{-+}^{-}$ of three field theory (\ref{a}) and plays the role
of the generator of the symmetry of this theory. It commutes also with
screening operators $Q_{i}={\displaystyle\oint}J_{i}(z)dz$. It is easy to check that field $G(z)$ commutes also with
screening operators $Q_{i}^{\prime}={\displaystyle\oint}I_{i}(z)dz$ generated by the currents
\begin{equation}
I_{1}(z)=\frac{1}{\sqrt{4N}}(i\alpha\partial_{z}\varphi_{3}+\alpha_{1}%
\partial_{z}\varphi_{1}+\alpha_{2}\partial_{z}\varphi_{2})\exp\left(
\frac{\varphi_{2}(z)}{2\alpha_{2}}\right)  \label{I1}%
\end{equation}
\begin{equation}
I_{2}(z)=\frac{1}{\sqrt{4N}}(i\alpha\partial_{z}\varphi_{3}+\alpha_{1}%
\partial_{z}\varphi_{1}+\alpha_{2}\partial_{z}\varphi_{2})\exp\left(
\frac{\varphi_{1}(z)}{2\alpha_{1}}\right)  \label{I2}%
\end{equation}
It is not surprising because currents $I_{i}$ differ from currents $J_{i}$ by
total derivatives which do not contribute to the screening operators. The main
advantage of currents $I_{i}$ is that their operator product expansion is not
singular. That permits us to write well defined action in the form
\begin{equation}
A_{12}=\int d^{2}x\left(  \frac{1}{16\pi}\sum_{i}^{3}(\partial_{\mu}%
\varphi_{i})^{2}-N\mu I_{1}\overline{I}_{1}-N\widetilde{\mu}I_{2}\overline
{I}_{2}+\frac{R^{\left(  2\right)  }}{8\pi}\Phi\right)  \label{A12}%
\end{equation}
with the same parameters $\mu,\widetilde{\mu}$ and dilaton field $\Phi.$ The
metric for sigma model corresponding to this action will have a form
\begin{align}
ds^{2}  &  =(1+\alpha^{2}U)dx_{0}^{2}+(1-\alpha_{1}^{2}U)dx_{1}^{2}%
+(1-\alpha_{1}^{2}U)dx_{2}^{2}-\nonumber\\
&  -2i\alpha\alpha_{1}Udx_{0}dx_{1}-2i\alpha\alpha_{2}Udx_{0}dx_{2}%
-2\alpha_{1}\alpha_{2}Udx_{1}dx_{2} \label{s12}%
\end{align}
where 
$ U(x)=me^{\frac{x_{2}}{2\alpha_{2}}}+\widetilde{m}e^{\frac{x_{1}%
}{2\alpha_{1}}}$
with $m=4\pi\mu,\widetilde{m}=4\pi\widetilde{\mu}$.

It can be checked that this metric satisfies the conformal invariance condition
\begin{equation}
R_{ik}+\nabla_{i}\nabla_{k}\Phi=0. \label{rr}%
\end{equation}
The curvature of this metric is 
\begin{equation}
 R=-\frac{1}{4}(me^{\frac{x_{2}}{2\alpha_{2}}%
}+\widetilde{m}e^{\frac{x_{1}}{2\alpha_{1}}})=-\frac{1}{4}U(x).
\end{equation}
It means that interaction is strong for large positive $x_{1,2}$.
It is convenient to modify equation (\ref{req}) for the wave function
describing zero modes dynamics and to write it in the form\footnote{The
application of this equation to the analysis  of the non-trivial Ricci flows
can be found in \cite{Fateev:1992tk,Fateev:1996ea,Perelman:fk,Fateev:2003sq}.}
\begin{equation}
\left(  -\frac{1}{\sqrt{g}}\partial_{i}\sqrt{g}g^{ik}\partial_{k}+\frac{R}%
{4}\right)  \Xi=(p_{1}^{2}+p_{2}^{2}+q^{2})\Xi.\label{Re}%
\end{equation}
The solution of this equation differs from the solution to eq (\ref{req}) by
the trivial factor $\exp[\Phi/2].$ The solution of our \ equation can be
represented in the form
\[
\Xi(x_{0},x_{1},x_{2})=e^{iqx_{0}}\Psi(x_{1},x_{2})
\]
where function $\Psi(x_{1},x_{2})$ satisfies to the following equation in the
partial derivatives ($\partial_{1}=\partial_{x_{1}};\partial_{1}%
=\partial_{x_{2}})$
\begin{align}
&  \left(  -\partial_{1}^{2}-\partial_{2}^{2}-U(x)\left(  (\alpha_{1}\partial
_{1}+\alpha_{2}\partial_{2})^{2}+\frac{\alpha_{1}}{2}(1-4\alpha q)\partial
_{1}+\frac{\alpha_{2}}{2}(1-4\alpha q)\partial_{2}\right)  \right)  \Psi
(x_{1},x_{2})\nonumber\\
&  =\left(\left(  \frac{\alpha q}{2}(4\alpha q-1)+\frac{1}{16}\right)U(x)+(p_{1}^{2}+p_{2}^{2})\right)  \Psi(x_{1},x_{2}).\label{pd}%
\end{align}
To solve our equation we substitute our function as
\begin{equation}
\Psi(x_{1},x_{2})=\frac{1}{(2\pi)^{2}}\int dq_{1}dq_{2}\widehat{\Psi}%
(q_{1},q_{2})\exp\left(  -iq_{1}x_{1}-iq_{2}x_{2}\right)  . \label{qq}%
\end{equation}
Then after simple transformation eq (\ref{pd}) can be written in the form
\begin{align}
&  \left(  2i\alpha_{1}q_{1}+2i\alpha_{2}q_{2}+2\alpha q+\frac{1}{2}\right)
^{2}\left(  \widetilde{m}\widehat{\Psi}(q_{1}-\frac{i}{2\alpha_{1}}%
,q_{2})+m\widehat{\Psi}(q_{1},q_{2}-\frac{i}{2\alpha_{2}})\right) \nonumber\\
&  =4\left(  (q_{1}^{2}-p_{1}^{2})+(q_{2}^{2}-p_{2}^{2})\right)  \widehat
{\Psi}(q_{1},q_{2}). \label{www}%
\end{align}
The unique solution to this equation corresponding to the regular function
$\Psi(x_{1},x_{2})$ at $x_{1,2}\rightarrow\infty$ is
\begin{align*}
\widehat{\Psi}(q_{1},q_{2})  &  =D(p_{1},p_{2},q)\frac{\Gamma(2i\alpha
_{1}(q_{1}-p_{1}))\Gamma(2i\alpha_{1}(q_{1}+p_{1}))}{\gamma(2i\alpha_{1}%
q_{1}+2i\alpha_{2}q_{2}+2\alpha q)}\times\\
&  \Gamma(2i\alpha_{2}(q_{2}-p_{2}))\Gamma(2i\alpha_{2}(q_{2}+p_{2}%
))(\widetilde{m}\alpha_{1}^{2})^{-i2\alpha_{1}q_{1}}(m\alpha_{2}%
^{2})^{-i2\alpha_{2}q_{2}}%
\end{align*}
where the constant
\[
D=(4\alpha_{1}\alpha_{2})(\widetilde{m}\alpha_{1}^{2})^{-i2\alpha_{1}p_{1}%
}(m\alpha_{2}^{2})^{-i2\alpha_{2}p_{2}}\frac{\gamma(-2i\alpha_{1}%
p_{1}-2i\alpha_{2}p_{2}+2\alpha q)}{\Gamma(-4i\alpha_{1}p_{1})\Gamma
(-4i\alpha_{2}p_{2})}%
\]
The contours of integration in the eq(\ref{qq}) go along the real axis
avoiding the poles at $q_{i}=\pm p_{i}$ counterclockwise. At $x_{1,2}%
\rightarrow-\infty$ (free region) we can close the contours in the upper half
plane. The main term of the asymptotic of the wave function will be defined
by the poles at $q_{i}=\pm p_{i}$ and will have a form
\begin{equation}
\Psi(x_{1},x_{2})=e^{ip_{1}x_{1}+ip_{2}x_{2}}+S_{1}e^{-ip_{1}x_{1}+ip_{2}%
x_{2}}+S_{2}e^{ip_{1}x_{1}-ip_{2}x_{2}}+S_{1,2}e^{-ip_{1}x_{1}-ip_{2}x_{2}%
}\label{S12}%
\end{equation}
where coefficients $S_{2}(p_{1},p_{2},q)$ and $S_{1}(p_{1},p_{2},q)$ were
calculated above and coefficient $S_{1,2}$ is
\begin{align*}
S_{1,2}(p_{1},p_{2},q)  & =(\widetilde{m}\alpha_{1}^{2})^{-i4\alpha_{1}p_{1}%
}(m\alpha_{2}^{2})^{-i4\alpha_{2}p_{2}}\frac{\Gamma(4i\alpha_{1}p_{1}%
)\Gamma(4i\alpha_{2}p_{2})}{\Gamma(-4i\alpha_{1}p_{1})\Gamma(-4i\alpha
_{2}p_{2})}\times\\
& \times\frac{\gamma(-2i\alpha_{1}p_{1}-2i\alpha_{2}p_{2}+2\alpha q)}%
{\gamma(+2i\alpha_{1}p_{1}+2i\alpha_{2}p_{2}+2\alpha q)}.
\end{align*}
If we take into account that $\widetilde{m}=4\pi\widetilde{\mu}$ and that
$\left(  \frac{\pi M}{4\alpha_{2}^{2}}\right)  =\left(  \pi\gamma
\widetilde{\mu}\left(  \frac{1}{4\alpha_{1}^{2}}\right)  \right)
^{2\alpha_{1}^{2}}\simeq\left(  \pi\widetilde{\mu}4\alpha_{1}^{2}\right)
^{2\alpha_{1}^{2}}$ we derive that $S_{1,2}$ coincides with minisuperspace
limit of reflection amplitude $R_{1,2}$ defined by eq (\ref{Rr}).
\bibliographystyle{MyStyle} 
\bibliography{MyBib}
\end{document}